\documentclass[reqno]{amsart}
\def\A{\mathcal A}

\def \Hil {\mathcal{H}}

\def \F {\mathcal{F}}

\def \I {\mathcal{I}}
\def \D {\mathcal{D}}
\def  \S{\mathcal{S}}
\def  \A{\mathcal{A}}
\def \R {\mathbb{R}}

\def \C {\mathbb{C}}
\def\u{\mathfrak{u}}

\def \Tr {\mathrm{Tr}}

\usepackage{xy}
\xyoption{all}

\def \Hil {\mathcal{H}}
\def \F {\mathcal{F}}
\def \I {\mathcal{I}}
\def \D {\mathcal{D}}
\def  \S{\mathcal{S}}
\def \R {\mathbb{R}}

\def \C {\mathbb{C}}
\def\u{\mathfrak{u}}

\def \Tr {\mathrm{Tr}}

\newcommand{\pd}[1]{\frac{\partial }{\partial #1}}
\newcommand{\pdc}[2]{\frac{\partial #1}{\partial #2}}
\def\pdd#1#2{\frac{\partial^2 #1}{\partial #2^2}}
\newcommand{\scalar}[1]{\langle #1 \rangle }

\newtheorem{proposition}{Proposition}

\newtheorem{lemma}{Lemma}
\newtheorem{example}{Example}
\newtheorem{remark}{Remark}

\begin{document}

%
%
\title{Reduction Procedures in Classical and Quantum Mechanics}

\author{Jos\'e F. Cari\~nena}
\address{Departamento de F\'\i sica Te\'orica \\ Universidad de Zaragoza \\ Pedro
  Cerbuna 12 \\ 50009 Zaragoza (SPAIN)}
\email{jfc@unizar.es}
\author{Jes\'us Clemente-Gallardo}
\address{Instituto de Biocomputaci\'on y F\'\i sica de los Sistemas Complejos \\
  Universidad de Zaragoza \\ Corona de Arag\'on 42 \\ 50009 Zaragoza (SPAIN)}
\email{jcg@unizar.es}
\author{Giuseppe Marmo}
\address{Dipartamento di F\'\i sica Teorica \\ Universit\'a Federico II and INFN
  sezione di Napoli
  \\ Via Cintia I\\ 81526 Napoli (ITALY)}
\email{marmo@na.infn.it}

\begin{abstract}
 We present, in a pedagogical style, many instances of reduction procedures
 appearing in a variety of physical situations, both classical and quantum. We
 concentrate on the essential aspects of any reduction procedure, both in the
 algebraic and geometrical setting, elucidating the analogies and the
 differences between the classical and the quantum situations.
\end{abstract}

\maketitle

\section{Introduction}

Reduction procedures \footnote{Expanded version of the Invited review talk
  delivered by G. Marmo at XXIst International 
  Workshop  On Differential Geometric Methods In Theoretical Mechanics, Madrid
  (Spain), September 2006 }, the way we understand them today (i.e in terms of Poisson
reduction) can be traced back to Sophus Lie in terms of function groups,
reciprocal function groups and indicial functions
\cite{Forsyth:1959,LieSche:1893,MSSV:1985}. Function groups provide an algebraic 
description of the cotangent bundle of a Lie group but are slightly more
general because can arise from Lie group actions which do not admit a momentum
map \cite{Mar:1983}. Recently they have reappeared as ``dual pairs''
\cite{How:1985}.

Physicists have used reduction procedures as an aid in trying to integrate the
dynamics ``by quadratures''. Dealing, as usual, with a coordinate formulation,
reduction and coordinate separability  have overlapped a good deal. From the
point of view of the integration of the equations of motion, the so called
decomposition into independent motions may be formalized as follows. Consider a
dynamical vector field $\Gamma$ on a carrier manifold $M$ and a decomposition 
$$
\Gamma=\sum_i\Gamma_i,
$$
with the requirement that:
\begin{itemize}
\item $[\Gamma_i, \Gamma_j]=0$
\item $\mathrm{span\,}\{ \Gamma_i(p)\}=T_pM\,,\quad \forall p\in X\subset M$, where $X$
  is open and dense submanifold in $M$.
\end{itemize}

When such a decomposition exists, the evolution is given by the product of the
one parameter groups associated with each $\Gamma_j$. 

Looking for superposition rules which would generalize the usual
superposition rule of linear systems, Lie \cite{LieSche:1893} introduced dynamical systems
admitting a decomposition 
$$
\Gamma=\sum_ia^j(t)\Gamma_j 
$$
with $[\Gamma_i, \Gamma_j]={\displaystyle{\sum_k}}c_{ij}^k\Gamma_k$ and $c_{ij}^k\in \mathbb{R}$
(i.e. the vector fields $\Gamma_k$ span a finite dimensional Lie algebra) and still
 $\{ \mathrm{span\,} \Gamma_i(p)\}=T_pM\,,
\ \forall p\in X\subset M$, where $X$ is open and dense in $M$. The
integration of these systems may be achieved by finding a fundamental set of
solutions: they admit a superposition rule even if the system is nonlinear.
These systems have been called \textbf{Lie Scheffers systems} (see e.g.
\cite{CarGrabMar:2000} and references therein)
and have an important representative given by the Riccati equation. It is
worth illustrating this example because it is an instance of a nonlinear
equation which is obtained as reduction of a linear one.

\begin{example}
  Let us consider $\R^2$ and the following system of first-order differential equation
  \begin{equation}
    \label{eq:r2}
    \frac d{dt}
\left ( \begin{array}{c} x_1 \\ x_2\end{array}\right )
=\left ( \begin{array}{cc} a_{11}&a_{12}\\a_{21}&a_{22}\end{array}\right )
\left ( \begin{array}{c} x_1 \\ x_2\end{array}\right )=
A
\left ( \begin{array}{c} x_1 \\ x_2\end{array}\right )
  \end{equation}
where $A$ is a $2\times 2$ matrix with real entries, maybe depending on
time. By performing a reduction with respect to the dilation group, or its
infinitesimal generator $\Delta=x_1\partial_{x_1}+x_2\partial_{x_2}$, i.e. by introducing
the variable $\xi=x_1/x_2$, the linear equation (\ref{eq:r2}) becomes:
\begin{equation}
  \label{eq:ricatti}
  \dot \xi=b_0+b_1\,\xi+b_2\xi^2,
\end{equation}
with 
$$b_0=a_{12}\,, \quad b_1=a_{11}-a_{22}\,,\quad b_2=-a_{21}\,.
$$

This is an instance of Riccati equation and is associated with a ``free''
motion on the group $\mathrm{SL}(2,\R)$: $\dot{g}\,g^{-1}=-b^{j}(t)A_{j}$, where
$\{A_j\}$ is a basis of the Lie algebra of $\mathrm{SL}(2,\R)$. Associated to it
we find a nonlinear superposition rule for the solutions: if $x_!,
x_2,x_3$ are independent solutions, every other solution $x$ is obtained from the
following ratio:
$$
\frac{(x-x_1)(x_2-x_3)}{(x-x_2)(x_1-x_3)}=K
$$
\end{example}

Riccati type equations arise also in the reduction of the Schr\"odinger
equation from the Hilbert space of states to the space of pure states
\cite{ChaErcMarMukSi:2007}.  

Another example, but for partial differential equations, is provided by the
following variant of the Burgers equation.

\begin{example}
To illustrate the procedure for partial differential equations in one space
and one time, we consider the following variant of the Burgers equation
$$\frac{\partial w}{\partial t}+ \frac{1}{2}\left( \frac{\partial w}{%
\partial x}\right) ^{2}-\frac{k}{2}\left( \frac{\partial ^{2}w}{\partial x^{2}%
}\right) =\allowbreak 0\,.$$

This equation admits a superposition rule of the following kind: for any two
solutions,  $w_{1}$ and $w_{2}$,
$$
w =-k\log \left( \exp \left( -\frac{w_{1}+ \ell _{1}}{k}\right) + \exp
\left( -\frac{w_{2}+ \ell _{2}}{k}\right) \right) 
$$
is again a solution with $\ell _{1}$ and $\ell _{2}$ and $k$ real constants.

The existence of a superposition rule might suggest that the equation may
be related to a linear one. This is indeed the case and we find that the heat
equation 

$$
\frac{\partial u}{\partial t}=\frac{k}{2}\frac{\partial ^{2}u}{\partial
x^{2}}
$$
is indeed related to the nonlinear equation by the replacement $u =\exp\left(
-\frac{w}{k}\right)$.
\end{example}
 
 Out of this experience one may consider the possibility of integrating more
 general evolution systems of differential equations by looking for a
 simpler system (simple here meaning that it is a system explicitly
 integrable) whose reduction  
gives the system that we would like to integrate. In some sense, with a
sentence, we could say that the reduction procedure provides us with
interacting systems out of free (or Harmonic) ones.

The great interest for new completely integrable systems boosted the
research in this direction in the past twenty five years and many
interesting physical systems, both in finite and infinite dimensions were
shown to arise in this way \cite{OlPere:1981}.

In the same ideology one may also put the attempts for the unification of
all the fundamental interactions in Nature by means of Kaluza-Klein
theories. In addition the attempt to quantize theories described by
degenerate Lagrangians called for a detailed analysis of reduction
procedures connected with constraints. These techniques came up again when
considering geometric quantization as a procedure to construct unitary
irreducible representations for Lie groups by means of the orbit method
\cite{Kir:1999}.

The simplest example to show how ``nonlinearities'' arise from reduction of
a free system is the three dimensional free particle. Of course if our
concern is primarily with the equations of motion we have to distinguish the
various available descriptions: Newtonian, Lagrangian, Hamiltonian. Each
description carries additional structures with respect to the equations of
motion and one has to decide whether the reduction should be performed within the
chosen category or if the reduced dynamics will be allowed to belong to
another one.

The present paper is a substantially revised version of a talk delivered at a
workshop. We have decided to keep the colloquial and friendly style aimed at
exhibiting the many instances of  reduction procedures appearing in a variety
of physical situations, both classical and quantum. 
This choice may give the impression of an episodic paper, however it contains
an illustration of the essential aspects of any reduction procedure ,both in
the algebraic and geometrical setting, pointing out the analogies and the
differences between  the classical and the quantum situation. Moreover it 
shows a basic philosophical principle: The unmanifest  world is simple and
linear, it is the manifest world which is ``folded'' and nonlinear. 

\subsection{Interacting systems from free ones}
In what follows, we are going to consider few examples where ``nonlinearities''
obtained from reduction of linear systems are carefully examined. 
\begin{example}
On $\R^{3}$ we consider the equations of motion of a free particle of unit mass in
Newtonian form: 
\begin{equation}
\ddot{\vec{r}}=0\,.  \label{eq:free}
\end{equation}

This system is associated to the second order vector field in $T{\R}^3$,  $\Gamma =\dot{%
\vec{r}}\frac{\partial }{\partial \vec{r}}$ and has constants of the motion 
$$
\frac{d}{dt}(\vec{r}\land \dot{\vec{r}})=0\,,\quad \frac{d}{dt}\dot{\vec{r}}=0\,.
$$

By introducing spherical polar coordinates 
$$
\vec{r}=r\,\vec{n}\qquad \vec{n}\cdot\vec{n}=1,\ r>0
$$
where $\vec n = \vec r /\| \vec r \|=\vec r /r$ is the unit vector in the
direction
 of ${\vec r}$,
and taking derivatives we find
$$\dot{\vec r} = \dot r\,\vec n + r\,\dot{\vec n}\,, \qquad \ddot{\vec r}
= \ddot r\,\vec n + 2\,\dot r \,\dot{\vec n} + r \,\ddot{\vec n}\ .$$

Moreover, from  the identities
$$
\vec n \cdot {\vec n} = 1\,, \qquad \vec n \cdot
\dot{\vec n} = 0, \qquad {\dot{\vec n}}^2 = -
\vec n \cdot {\ddot{\vec n}}\ ,
$$
we see that 
 $\dot{\vec r}\cdot \vec n=\dot r$, and  $ {\vec r}\cdot \dot{\vec r}=r\,\dot r$;
 using
$\ddot {\vec r}=0$ we obtain
\begin{equation}\label{radial_3}   \ddot r = - r\,\vec n \cdot\ddot{\vec n} = r\
\dot{\vec n}^2\, ,
\end{equation}
and, of course, 
$$ 
\vec{r}%
\land \dot{\vec{r}}=r^{2}\vec{n}\land \dot{\vec{n}}\, .
$$

The equations of motion (\ref{radial_3}) are not equations in the
variable $r$ only, because of the term ${\dot{\vec n}}^2$.
However by making use of constants of the motion, we can choose
invariant submanifolds $\Sigma$ for $\Gamma$ such that taking the restrictions on such submanifolds,
 we can associate with this equation an equation of motion
involving only $r, \dot r$ and some ``coupling constants'' related 
to the values of the constants of motion. So, we can
 restrict ourselves to initial
conditions with a fixed value of the angular momentum, say, for instance, 
$$
l^{2}=r^{4}(\dot{\vec{n}})^{2},
$$
in order to get 
$$
\ddot{r}=\frac{l^{2}}{r^{3}}.
$$

If, on the other hand, we restrict ourselves to initial conditions
satisfying 
$$
(\dot{\vec{r}})^{2}=2E,
$$
we get 
$$
\ddot{r}=\frac{2E}{r}-\frac{\dot{r}^{2}}{r}.
$$

By selecting an invariant submanifold of $\R^{3}$ by means of a convex
combination of energy and angular momentum, i.e. $\alpha ({\vec r} \land \dot {\vec r} )^2 +
(1-\alpha ) \,\dot{\vec r}^2  = k$,
we would find 
$$
\ddot{r}=\left ( \frac{\alpha l^2+(1-\alpha)(2E-\dot r^2)r^2}{r^3} \right )
$$
\end{example}

We might even select a time dependent constant of the motion,
for instance
$$  r^2 + \dot{\vec r}^2\, t^2 - 2\, {\vec r} \cdot \dot{\vec r}\ t = k^2\ ,
$$
 to get rid of $(\dot{\vec{n}})^{2}$,
$$
(\dot{\vec{n}})^{2}= \frac{1}{r^2}\  [(k^2 + 2\,{\>r}
\cdot \dot{\>r}\ t - r^2)\,t^{-2} - \dot r^2 ]\ 
$$
 and thus we would get a time-dependent reduced dynamics:
$$ \ddot r = \frac{k^2}{r}\, t^{-2} + 2\,\dot r\, t^{-1} -
\frac{t^{-2}}{r} - \frac{\dot r^2}{r}\ .$$

The geometrical interpretation of what we have done is rather simple: we
have selected an invariant submanifold $\Sigma \subset \R^{3}$ (the level
set of a constant of the motion), we have restricted the dynamics to it, and
then we have used the rotation group to foliate $\Sigma $ into orbits. The
reduced dynamics is a vector field acting on the space of orbits $\widetilde{%
\Sigma}=\Sigma /SO(3)$. It should be remarked that even if $\Sigma $ is
selected in various ways, the choice we have made is compatible with the
action of the rotation group. It should be clear now that our presentation
goes beyond the standard reduction in terms of the momentum map, which involves
additional structures. Indeed this
reduction, when carried out with the canonical symplectic structure,  would
give us only the first solution in the example above. 

There is another way to undertake the reduction. On $T^*\R^3$ with coordinates 
$(\vec r, \vec p)$, we can consider the functions 
$$
\xi_1=\frac 12 \langle \vec r, \vec r\rangle\,, \qquad \xi_2=\langle \vec p,
\vec p\rangle\,, \qquad \xi_3=\langle \vec r, \vec p\rangle\,.
$$
Here $\langle \vec a, \vec b\rangle$ denotes the scalar product $\vec a\cdot
\vec b$, but it can be extended to a non definite positive scalar product.

The equation of motion (\ref{eq:free}) on these coordinate functions becomes 
$$
\frac d{dt}\xi_1=\xi_3\,, \qquad \frac d{dt}\xi_2=0\,, \qquad \frac
d{dt}\xi_3=\xi_2. 
$$
Note that any constant of the motion of this system is then a function of $\xi_2$ and
$(2\xi_1\xi_2-\xi_3^2)$. Consider first the invariant submanifold $\xi_2=k\in
\R$. Then we find, 

$$
\frac{d}{dt}\xi _{1}=\xi _{3}\,,\qquad \frac{d}{dt}\xi _{3}=k\,,
$$
i.e. a uniformly accelerated motion in the variable $\xi_1$. It may be described
by the Lagrangian $L= 
\frac{1}{2}v^{2}+ kx$, where $x=\xi_1$, $v=\dot \xi_1=\xi_3$.

Had we selected a different  invariant submanifold, for instance, 
$$
2\xi _{1}\xi _{2}-\xi _{3}^{2}=l^{2},
$$
the restricted dynamics would have been: 
$$
\frac{d}{dt}\xi _{1}=\xi _{3}\,,\qquad \frac{d}{dt}\xi _{3}=\frac{\xi
_{3}^{2}+l^{2}}{2\xi _{1}}\,.
$$

A corresponding Lagrangian description is provided by the function 
 $L=\frac{1}{2}\frac{v^{2}}{%
x}-\frac{2l^{2}}{x}$, with $x=\xi_1$ and $\dot x=v=\xi_3$.

If we start with the dynamics of the isotropic harmonic oscillator, say
$\dot{\vec r}=\vec p$ and $\dot{\vec p}=-\vec r$, on functions
$\eta_1=\xi_1-\frac 12 \xi_2$, $\xi_3$ and $\eta_2=\xi_1+\xi_2$, we would get
$\dot \eta_1=2\xi_3$, $\dot \xi_3=2\xi_1-\xi_2$ and $\dot \eta_2=0$, i.e. $\dot
\eta_1=2\xi_3$ and $\dot \xi_3=-2\eta_1$, i.e. we get a one dimensional oscillator. We
would like to stress that the ``position'' of this reduced system, say $\eta_1$ is not
a function depending only on the initial position variables.

\begin{remark}
  Let us point out a general aspect of the example we just
  considered. We first notice that the functions $\xi_1=\frac 12 x_ax^a$,
  $\xi_2=p_ap^a$ and $\xi_3=x_ap^a$ may be defined on any phase space
  $\R^{2n}=T^*\R^n$, with $\R^n$ an Euclidean space. If we consider the
  standard Poisson bracket, say 
$$
\{ p_a, x^b\} =\delta_a^b, \quad \{p_a, p_b\} =0=\{ x^a, x^b\}, 
$$
we find that for the new variables
\begin{equation}
  \label{eq:pbr}
\{ \xi_3, \xi_1\} =2\xi_1, \quad \{ \xi_2, \xi_3\} =2\xi_2, \quad \{ \xi_2, \xi_1\} =2\xi_3.  
\end{equation}

Thus the functions we are considering close on  the Lie algebra
$\mathfrak{sl}(2, \R)$.  The infinitesimal generators $\{\xi_i, \cdot \}$ are
complete vector fields and integrate to a symplectic action of $SL(2, \R) $ on
$\R^{2n}$. 

Then, in the stated conditions there is always a symplectic action of
$SL(2,\R)$ on $T^*\mathbb{R}^n\simeq\mathbb{R}^{2n}$ with a corresponding momentum
map $\mu:T\R^n\to \mathfrak{sl}^*(2, \R)$. If we denote again the coordinate
functions on this three dimensional vector space by 
$\{\xi_1, \xi_2, \xi_3\} $, we have the Poisson bracket  (\ref{eq:pbr}) and the
momentum map provides a symplectic realization of the Poisson manifold $
\mathfrak{sl}^*(2, \R)$. In the language of Lie, the coordinate functions
$\{\xi_1, \xi_2, \xi_3\} $ along with all the smooth functions of them $\{ f(\xi_1, \xi_2,
\xi_3)\} $ define a function group. The Poisson subalgebra of functions of
$\F(\R^3)$ commuting with all the functions $f(\xi_1, \xi_2, \xi_3)$, constitute the
reciprocal function group, and all functions in the intersection of both sets,
say functions of the form $F(\xi_1\xi_2 -\frac 12 \xi_3^2)$, constitute the indicial
functions.

By setting $\xi_1=\frac 12, \xi_3=0$ we identify a submanifold in $T\R^n$
diffeomorphic with $TS^{n-1}$, the tangent bundle of the $(n-1)$-dimensional
sphere. It is clear that the reciprocal function group  is generated by
functions $J_{ab}=p_ax^b-p_bx^a$. Thus, the reduced dynamics which we usually
associate with the Hamiltonian $H=\frac 12 p_r^2+\frac{l^2}{2r^2}$ is actually
a dynamics on $\mathfrak{sl}^* (2, \R)$ and therefore it has the same form
independently of the dimension of the space $T\R^n$ we start with. Symplectic
leaves in $\mathfrak{sl}^*(2, \R)$ are diffeomorphic to $\R^2$ and pairs of conjugated
variables may be introduced as 
\begin{eqnarray*}
\left \{ \frac{\xi_3}{\xi_1}, \frac 12 \xi_1\right \} &=1 \mathrm{\quad  or \quad}
\left \{  \frac 12 \xi_2, \frac{\xi_3}{\xi_2}\right \} =1 \mathrm{\quad or \quad} \\
\left \{ \frac{\xi_3}{\sqrt{2\xi_1}}, \sqrt{2 \xi_1}\right \} &=1 \mathrm{\quad  or \quad}
\left \{ \frac{\xi_3}{\sqrt{2\xi_2}}, \sqrt{2 \xi_2}\right \} =1.
\end{eqnarray*}

\end{remark}

We see in all these examples that the chosen invariant submanifold appears
eventually as a ``coupling constant'' in the reduced dynamics. Moreover, the final
``second order description'' may be completely unrelated to the original one.

Another remark is in order. We have not specified the signature of our
scalar product on $\R^{3}$. It is important to notice that the final result
does not depend on it. However, because in the reduced dynamics $\xi _{1}$
appears in the denominator, when the scalar product is not positive definite
we have to remove the full algebraic variety $\langle \vec{r},\vec{r}\rangle
=0$ to get a smooth vector field. If the signature is ($+ ,+ ,-),$ 
the relevant group will not be SO(3) anymore but will be replaced by
 $SO(2,1)$.

We can summarize by saying that  the reduction of the various examples that we have
considered are based on the selection of an invariant submanifold and the
selection of an invariant subalgebra of functions.

A few more remarks are necessary:

\begin{remark}
If we consider the Lagrangian description of the free particle 
$$
{L}=\frac{1}{2}\langle \dot{\vec{r}},\dot{\vec{r}}\rangle ,
$$
in polar coordinates it becomes 
$$
{L}=\frac{1}{2}\left( \dot{r}^{2}+r^{2}(\dot{\vec{n}})^{2}\right) ,
$$
which restricted to the submanifold $l^{2}=r^{4}(\dot{\vec{n}})^{2}$ would
give 
$$
{L}=\frac{1}{2}\left( \dot{r}^{2}+\frac{l^{2}}{r^{2}}\right) ,
$$
which is not
 the Lagrangian giving rise to the dynamics $\ddot{r}=l^{2}/
{r^{3}}$. Therefore, we must conclude that the reduction, if done in the
Lagrangian formalism,  must be considered
as a symplectic reduction in terms of the symplectic structure of Lagrangian
systems (i.e. in terms of the symplectic form $\omega _{{L}}$ and
the energy function $E_{{L}}$).
\end{remark}

\begin{remark}
The free particle admits many alternative Lagrangians, therefore once an
invariant submanifold $\Sigma $ has been selected, we have many alternative
symplectic structures to pull-back to $\Sigma $ and define alternative
involutive distributions to quotient $\Sigma $. The possibility of endowing
the quotient manifold with a tangent bundle structure has to be investigated
separately because the invariant submanifold $\Sigma $ does not need to have
a particular behaviour with respect to the tangent bundle structure.
A recent generalization consists of considering that the quotient space may not
have a tangent bundle structure but may have  a Lie algebroid
structure. Further examples and  additional comments on previous
examples may be found in \cite{LanMarSpaVi:1991,ManMar:1992}. 
\end{remark}

We shall close now these preliminaries by commenting on the generalization
of this procedure to free systems on higher dimensional spaces.

\subsection{Generalized polar coordinates}

In the existing literature, examples have been already considered to get
Calogero-Moser potentials, Toda 
and other well-known systems, starting with free or harmonic motions on the
space of $n\times n$ 
Hermitian matrices, free motions on $U(n)$, and free motions on the coset space 
$GL(n, \C)/U(n, \C)$ \cite{OlPere:1981}.
 
The main idea is to start with a space of diagonalizable matrices $\{ X\}$
and to consider a diagonalizing matrix $G$ in such a way that 
$$
X=GQG^{-1}. 
$$

The diagonal matrix will play the role of ``radial coordinates'' while $G$
plays the role of angular coordinates. Some care is needed when the
parametrization of $X$, $X=X(Q, G)$, is not unique.

\begin{example}
Let us study again the three dimensional example discussed above. Consider matrices 
$$
X=
\left ( \begin{array}{cc}
x_1 & \frac{x_2}{\sqrt{2}} \\
\frac{x_2}{\sqrt{2}} & x_3
\end{array}\right ),
$$
satisfying  the evolution equation
$$ \ddot{X}=0\,.
$$
Therefore, the matrix $M=[X,\dot{X}]$ is such that 
$ {dM}/{dt}=0$, because 
$$\dot M = [\dot X, \dot X ] + [X, \ddot X ] = 0\ .$$

We can introduce  new coordinates for the symmetric matrix $X$ by using the
rotation group:  such 
a matrix $X$ can be diagonalized by means of  an orthogonal
transformation $G$, thus, $X$ can be written as
$X=GQG^{-1}$ with 
$$
Q=\left ( \begin{array}{cc}q_1 & 0 \\ 0 &q_2 \end{array}\right )\,, \qquad G=
\left ( \begin{array}{cc}\cos \phi & \sin \phi \\ -\sin \phi &
  \cos\phi \end{array}\right )\quad
$$
and therefore, as 
$$
G\,Q\, G^{-1}=
\left (
\begin{array}{cccc}
q_1 \cos^2\phi+q_2\,\sin^2\phi &&&(q_2-q_1)\,
\sin\phi\,\cos\phi\\
(q_2-q_1)\,\sin\phi\,\cos\phi &&&
q_1\,\sin^2\phi+q_2\,\cos^2\phi
\end{array}
\right )
$$
and $x_1=q_1\cos^2\phi+q_2\sin^2\phi $, $x_3=q_1\sin^2\phi+q_2\cos^2\phi$,
we get the relation 
$$x_1+x_3  = q_1 + q_2 \,, \qquad x_2 = \frac{1}{\sqrt 2}
(q_2 - q_1) \sin 2\phi \,, \qquad x_1-x_3 =  (q_1 - q_2)
\cos 2\varphi\ .$$

Note that $ G^{-1}\dot{G}= \dot{G} \,G^{-1}=\dot{\phi}\sigma $ and
$G\,\sigma=\sigma\, G$, where
$$
 \sigma =i\, \sigma_2=\left ( \begin{array}{cc} 0 & 1 \\ -1 & 0 \end{array}\right )\,,
$$
and $M=l\,\sigma$, where $l$ is the sum of the first and third components of the
angular momentum.

Then, using  
$$\frac d {dt}G^{-1}=-G^{-1}\,\dot G\, G^{-1}\,,
$$
we see that 
$$\dot{X} = \dot G\, Q\,G^{-1}-G\,Q\,G^{-1}\,\dot G\, G^{-1}+G\,\dot Q\,
G^{-1}=
G \left( [G^{-1}\dot G, Q] + \dot Q \right) G^{-1}  ,$$
i.e.
$$
\dot{X} = G(\dot Q+\dot\phi \,[\sigma,Q])\,G^{-1}\,.
$$

Notice that 
$[\sigma,Q]=(q_2-q_1)\,\sigma_1$ and $[Q,\dot Q]=0$. Consequently, 
$$M=[X,\dot X]= G\,[Q,\dot\phi \,[\sigma,Q]+\dot Q]\,G^{-1}=\dot\phi \, 
(q_2-q_1)G\,[Q,\sigma]\,G^{-1}=-\dot\phi \, 
(q_2-q_1)^2\,\sigma\,,
$$
and then $l$ is given by 
$$l=\dot\phi \, 
(q_2-q_1)^2\,.
$$

The equations of motion along the radial variables become 
$$
\ddot{Q}-\dot{\phi}^{2}[\sigma ,[\sigma ,Q]]=0
$$

Restricting to the submanifold $\Sigma _{l}$
given by
$$
\Sigma _{l}=\left\{l=-\frac{1}{2}\Tr M\sigma \right\}\,.
$$
 we find 
$$
\ddot{Q}=\frac{l^{2}}{(q_{2}-q_{1})^{4}}[\sigma ,[\sigma ,Q]]
$$
or, more explicitly 
$$
\ddot{q}_{1}=-\frac{2l^{2}}{(q_{2}-q_{1})^{3}}\quad \ddot{q}_{2}=\frac{2l^{2}%
}{(q_{2}-q_{1})^{3}}
$$
They provide us with Calogero equations for two interacting particles on a line
and  are the Euler--Lagrange equations associated with the
Lagrangian function,
$$L = \frac 12 \left( \dot q_1^2 - \dot q_2^2 \right) - \frac{g^2}{(q_2
- q_1)^2}\ .$$

\end{example}

\subsection{A Lagrangian description and solutions of the Hamilton-Jacobi equation}
\label{sec:lagr-descr-solut}
On the space of symmetric matrices $\{ X\} $ we define the Lagrangian function 
$L=\frac 12 \Tr (\dot X \dot X)$. This Lagrangian gives rise to the
Euler-Lagrange equations of motion $\ddot X=0$.
 Moreover, the symplectic
structure associated to it is defined by $\omega_L=\Tr (d\dot X\land dX)$ and
$E_L=L$. The invariance of the Lagrangian under translations and rotations
implies the conservation of linear momentum $P=\dot X$ and angular momentum
$M=[X, \dot X]$. The corresponding explicit solutions of the dynamics are thus
given by
$$
X(t)=X_0+tP_0.
$$

It is possible to find easily a solution of the corresponding Hamilton-Jacobi
equation. Indeed,
by integrating the Lagrangian along the solutions or by solving
$P_tdX_t-P_0dX_0=dS(X_t, X_0;t)$ with $P_t=P_0=P$ and $P=t^{-1}(X(t)-X_0)$, we
find that the action is written as $S=\frac 1{2t}\Tr(X_t-X_0)^2$.

By fixing a value $\ell^2=\frac 12 \Tr M^2$ we select an invariant manifold
$\Sigma$. The corresponding reduced dynamics gives the Calogero equations.
 Therefore we restrict
$S$ to those solutions which satisfy $\frac 12 \Tr (X_t^2X_0^2-(X_tX_0)^2)=\ell^2$
and we find a solution for the Hamilton-Jacobi equation associated with the
reduced dynamics.

\begin{remark}
  For any invertible symmetric matrix $K$, the Lagrangian function $L_K=\frac
  12 \Tr \dot X K \dot X$ would describe again the free motion. More generally,
  for any monotonic function $f$, the composition $f(L_K)$ would be a possible
  alternative Lagrangian. The corresponding Lagrangian symplectic structure
  could be used to find alternative Hamilton-Jacobi equations. For those
  aspects we refer to \cite{CaGraMarMarMunRom:2006}.
\end{remark}

\section{Summarizing and formalizing}

\bigskip To prepare the ground for the Poisson reduction we emphasize that
the reduction procedure that we have considered so far uses two basic
ingredients:

\begin{itemize}
\item  An invariant subalgebra (of functions) $R$.

\item  An invariant submanifold of the carrier space $\Sigma \subset M$.
\end{itemize}

\subsection{The geometrical description}

Let us try to identify the basic aspects of the reduction procedures we shall consider.
We denote by $M$ the manifold  containing the states of our 
system. The equations of motion will be represented by a vector field $\Gamma
$, and we suppose that it gives rise to a one parameter group of
diffeomorphisms 
$$
\Phi:\R\times M\to M \,.
$$
Occasionally, when we want our map to keep track of the infinitesimal generator we
will write $\Phi_\Gamma$ or $\Phi_\Gamma:\mathbb{R}\times M\to M$.

To apply the  general reduction procedure we need:

\begin{itemize}
\item  a submanifold $\Sigma $, invariant under the $\Phi $ evolution, i.e. 
$$
\Phi (\mathbb{R}\times \Sigma )\subset \Sigma\,, \mathrm{\quad or \quad}\Phi (t,m)\in \Sigma\,, \ \forall t\in \R,
\,\,m\in \Sigma \,.
$$

\item  An invariant equivalence relation among points of $\Sigma $, i.e. we consider
  equivalence relations for which 
$$
m\sim m^{\prime }\Rightarrow \Phi (\mathbb{R},m)\sim \Phi (\mathbb{R},m^{\prime })\,.
$$
\end{itemize}

The reduced dynamics or ``reduced evolution'' is defined on the manifold of
equivalence classes (assumed to be endowed with a differentiable structure).

One may also start the other way around: we could first consider an
invariant equivalence relation on the whole manifold $M$ and then select an
invariant submanifold for the reduced dynamics, to further reduce the
dynamical evolution.

\subsubsection{Some remarks}

In real physical situations the invariant submanifolds arise as level set of
functions. These level sets were called \textbf{invariant relations} by Levi-Civita
\cite{AmalLev:1926} to distinguish them from level sets of constants of the
motion. Usually, 
equivalence classes will be generated by orbits of Lie groups or leaves of
involutive distributions. ``Closed subgroup'' theorems are often
employed to guarantee the regularity of the quotient manifold  \cite{Palais:1957}.

When additional structures are present, like Poisson or symplectic
structures, it is possible to get involutive distributions out of a family
of invariant relations. The so called ``symplectic reduction'' is an example
of this particular situation.

When the space is endowed with additional structures, say a
tangent or a cotangent bundle, with the starting dynamics being, for
instance, second order (in the tangent case), we may also ask for the reduced
one to be second 
order, once we ask the reduced space to be also endowed with a tangent space
structure. This raises natural questions on how to find appropriate tangent
or cotangent bundle structures on a given manifold obtained as a reduced
carrier space. Similarly, we may start with a linear dynamics, perform a
reduction procedure (perhaps by means of quadratic invariant relations) and
enquire on possible linear structures on the reduced carrier space.
A simple example of this situation is provided by the Maxwell equations. These
equations may be written in terms of the Faraday  $2$--form $F$ encoding the
electric field $E$ and the magnetic field $B$, as:
$$
dF=0 \quad d *F=0,
$$
when considered in the vacuum \cite{MarParaTulcz:2005}. We may restrict
these equations to the invariant submanifold 
$$
F\land F=0,  F\land *F=0.
$$

Even though these relations are quadratic the reduced Maxwell equations provide
as solutions the radiation fields and are still linear.

In conclusion, when additional structures are brought into the picture, we
may end up with extremely rich mathematical structures and quite difficult
mathematical problems.
\begin{example}
{\bf A charged non-relativistic particle in a magnetic monopole field}

This system was considered by Dirac \cite{Dirac:1931} and a variant of it, earlier by
Poincar\'e \cite{Poin:1896}. To describe it in terms of a
Lagrangian Dirac introduced a ``Dirac string''. The presence of this unphysical
singularity leads to technical difficulties in the quantization of this
system. Several proposals have been made to deal with these problems. 

Here we would like to show how our reduction procedure allows to deal with this
system and provides a clear way for its quantization. In doing this we shall
follow mainly \cite{BaMaSkaSter:1980,BaMaSkaSter:1983,BaMaSkaSter:1991}

The main idea is to replace $\R_0^3$ with $\R_0^4$ described  as the product
$\R_0^4=S^3\times \R^+$, and to get back our space of relative coordinates for
the charge-monopole by means of a reduction procedure.

We set first $\vec x\cdot\vec \sigma=rs\sigma_3s^{-1}$, where
$r^2=x_1^2+x_2^2+x_3^2$ and $s\in SU(2)$ (realized as $2\times 2$ matrices of
the defining representation; while $\{ \sigma_1, \sigma_2,\sigma_3\} $ are the
Pauli matrices. We write the Lagrangian function on $\R_0^4$ as
$$
L=\frac 12 m \Tr \left (\frac d{dt}(rs\sigma_3s^{-1})\right )^2-k(\Tr
\sigma_3s^{-1}\dot s)^2.
$$ 

This expression for the Lagrangian shows clearly the invariance under the left
action of $SU(2)$ on itself and an additional invariance under the right action
$s\mapsto se^{i\sigma_3\theta}$ for $\theta\in [0,2\pi)$. It is convenient to
introduce left invariant one forms $\theta^a$ by means of  $i\sigma_a \theta ^a=s^{-1}ds$
and related left invariant
vector fields $X_a$ which are dual to them $\theta^a(X_b)=\delta^a_b$. If
$\Gamma$ denotes any second order vector field on $\R_0^4$ we set $\dot
\theta^a=\theta^a(\Gamma)$, where, with some abuse of notation, we are using
the same symbol for $\theta^a$ on $\R_0^4$ and its pull-back to $T\R_0^4$. It
is also convenient to use the unit vector $\vec n$ defined by ${\vec x}=
{\vec n}r$, i.e $\vec n \vec \sigma=s\sigma_3s^{-1}$.

After some computations, the Lagrangian becomes
\begin{equation*}
  L=\frac 12 m\dot r^2+\frac 14 mr^2(\dot \theta_1^2+\dot
  \theta_2^2)+k\dot \theta_3^2.
\end{equation*}

It is not difficult to find the canonical $1$-- and $2$--forms for the Lagrangian
symplectic structure. For instance $\theta_L=m\dot rdr+\frac 12 mr^2(\dot
\theta_1\theta_1+\dot \theta_2\theta_2)+2k\dot \theta_3\theta_3$; and of
course $\omega_L=d\theta_L$. The energy function $E_L$ coincides with $L$.

If we fix the submanifold $\Sigma_c$ by setting 
$$
\Sigma_c=\{ ((r,v)\in T\R_0^4 \mid  \dot \theta_3=c\} ,
$$
the submanifold turns out to be invariant because $\dot \theta_3$ is a constant
of the motion.

On $\Sigma_c$, $\theta_L=m\dot r dr+\frac 12 mr^2(\dot
\theta_1\theta_1+\dot \theta_2\theta_2)+2kc\theta_3$. If we then use the
foliation associated with $X_3^T$ (the tangent lift of $X_3$ to $T\R_0^4$), we
find that $\omega_L$ is the pull-back of a 2--form on the quotient because $d\theta_3=\theta_1\land
\theta_2$, and hence contains $X_3^T$ in its kernel. The term $d\theta_3$ is
exactly proportional to the magnetic field of the magnetic monopole sitting at
the origin. Thus on the quotient space of $\Sigma_c$ by the action of the left
flow of $X_3^T$ we recover the dynamics of the electron-monopole system on the
(quotient) space $T(S^2\times \R)=T\R_0^3$. It is not difficult to show that 
\begin{equation*}
  \frac d {dt} \left ( -\frac i2 [\vec n\vec \sigma, mr^2\dot{\vec
      n}\vec \sigma]+k\vec n \vec \sigma \right )=0; \quad k=\frac {eg}{4\pi}.
\end{equation*}
These constants of the motion are associated with the rotational invariance and
replace the usual angular momentum functions.
\end{example}

This example shows that the reduction of the Lagrangian system of Kaluza-Klein
type on $T\R^4$  does not reduce to a Lagrangian system on $T\R^3$ but just to
a symplectic system.

\subsection{The algebraic description}

The evaluation map $ev:M\times \F\to \R$ defined as $(m,f)\mapsto f(m)$,
allows to dualize the basic ingredients from the manifold to the algebra of
functions on $M$, the observables.

We first notice that to any submanifold $\Sigma \subset M$ we can associate
a short exact sequence of associative algebras 
$$\xymatrix{0\ar[r]&\I_\Sigma\ar[r]&\F\ar[r]^{\pi_\Sigma }&\F_\Sigma\ar[r]&0}
$$
defined in terms of the identification map $i_\Sigma:\Sigma \hookrightarrow M
$, $\Sigma\ni m\mapsto m\in M$. We have thus: 
\begin{equation*}  
I_\Sigma=\left \{\, f\in \F \mid i_\Sigma^*(f)=0\,\right \}
\end{equation*}

From the property $i_\Sigma^*(fg)=i_\Sigma^*(f)i_\Sigma^*(g)$ we find that $%
\I_\Sigma$ is a bilateral ideal in $\F$. The algebra $\F_\Sigma$ is then the
quotient algebra $\F/\I_\Sigma$.

Any derivation $\Gamma$ acting on the set of functions of $\F(M)$ will
define a derivation on the set of functions $\F_\Sigma$ if and only if $%
L_\Gamma \I_\Sigma \subset \I_\Sigma$, so that $\Gamma$ acting on
equivalence classes will define a derivation on the reduced carrier space.

A simple example illustrates the procedure. On $T\R^3$ we consider the
bilateral ideal $\I_\Sigma$, when $\Sigma$ is defined from 

\begin{equation}
  \label{eq:constr}
f_1=\vec r \cdot \vec r-1 \quad f_2=\vec r \cdot \vec v,  
\end{equation}
and we set $f_1=0=f_2$.

We get the submanifold $\Sigma$ as $TS^2$. The algebra of functions $\F_\Sigma$ is
obtained from $\F(T\R^3)$ simply by using in the argument of $f(\vec r, \vec
v)$ the constraints  (\ref{eq:constr}), i.e. $f_1=0=f_2$.  A vector field $X$
on $T\R^3$ will be tangent to $\Sigma=TS^2$  if and only if 
$$
L_X(\vec r\cdot  \vec r-1)=\alpha (\vec r \cdot \vec r-1)+\beta \vec r \cdot \vec v,
$$
for arbitrary functions $\alpha, \beta$ and also 
$$
L_X(\vec r\cdot  \vec v)=\alpha' (\vec r \cdot \vec r-1)+\beta' \vec r \cdot \vec v.
$$

It is not difficult to show that the module of such derivations is generated by 
$$
\mathcal{R}_l=\epsilon_{jkl}\left (x_j\pd{x_k}+v_j\pd{v_k}\right ); \quad \mathcal{V}_l=\epsilon_{lij} x_j\pd{v_i}
$$

An invariant subalgebra in $\F$, say $\widetilde{\F}$, for which $\I_\Sigma$
is an ideal, defines an invariant equivalence relation by setting 
\begin{equation}  \label{eq:equi}
m^{\prime}\sim m^{\prime\prime}\mathrm{\quad iff \quad } f(m^{\prime})=f(m^{\prime%
\prime}), \quad \forall f\in \widetilde{\F}
\end{equation}

It follows that $\widetilde{\F}$ defines a subalgebra in $\F_\Sigma$ and
corresponds to a possible quotient manifold of $\Sigma$ by the equivalence
relation defined by $\widetilde{\F}$.

In general, a subalgebra in $\F$, say $\F_Q$, defines a short exact
sequence of Lie algebras
\begin{equation}  \label{eq:subalshort}
0\longrightarrow \mathfrak{X}^v\longrightarrow \mathfrak{X}^N\longrightarrow %
\mathfrak{X}_Q\longrightarrow 0
\end{equation}
where $\mathfrak{X}^v$ is the Lie algebra of vector fields annihilating $\F_Q
$, $\mathfrak{X}^N$ is the normalizer of $\mathfrak{X}^v$ in $\mathfrak{X}(M)
$, and $\mathfrak{X}_Q $ is the quotient Lie algebra. This sequence of Lie
algebras may be considered a sequence of Lie modules with coefficients in $\F%
_Q$. In the previous case, $\F_Q$ would be the invariant subalgebra in $\F%
_\Sigma $ and the equivalence relation would be defined by the leaves of the
involutive distribution $\mathfrak{X}^v$ (regularity requirements should be
then imposed on $\F_Q$). See \cite{LanMar:1990} for details.

From the dual point of view it is now clear that reducible evolutions will
be defined by one-parameter groups of transformations which are
automorphisms of the corresponding short exact sequences. The corresponding
infinitesimal versions will be defined in terms of derivations of the
appropriate short exact sequence of algebras.

To illustrate this aspect, we consider the associative subalgebra of
$\F(T\R^3)$ generated by $\{ \vec r \cdot \vec r, \vec v\cdot \vec v, \vec r \cdot \vec v\} $. For
this algebra it is not difficult to see that the vector fields 
$$
X_c=\epsilon_{abc}\left(   x^a\pd{x_b}+v^a\pd{v_b} \right )
$$
generate $\mathfrak{X}^v$, while $\mathfrak{X}^N$ is generated by $\mathfrak{X}^v$ and 
$$
\vec r \pd{\vec v}, \quad \vec v\pd{\vec r}, \quad \vec r\pd{\vec r}, \quad  \vec v
\pd{\vec v}.
$$
The quotient  $\mathfrak{X}_Q$, with a slight abuse of notation, can also be
considered to be  generated by the vector fields
$$
\vec r \pd{\vec v}, \quad \vec v\pd{\vec r}, \quad \vec r\pd{\vec r}, \quad \vec v
\pd{\vec v},
$$
which however are not all independent. Any combination of them with
coefficients in the subalgebra may be considered a ``reduced dynamics''.

\subsection{Additional structures: Poisson reduction}

When a Poisson structure is available, we can further qualify the previous
picture. We can consider associated short exact sequences of Hamiltonian
derivations. Hence, a Poisson reduction can be formulated in the following
way: we start with $\I_{\Sigma }$, again an ideal in the commutative and
associative algebra $\F$. We consider then the Hamiltonian derivations which
map $\I_{\Sigma }$ into itself: 
$$
W(\I_{\Sigma })=\{f\in \F\,  \mid \{f,\I_{\Sigma }\}\subset \I_{\Sigma }\}\,.
$$
Then we consider $\I_{\Sigma }^{\prime }=I_{\Sigma }\cap W(\I_{\Sigma })$
and get the exact sequence of Poisson algebras 
$$
0\longrightarrow \I_{\Sigma }^{\prime }\longrightarrow W(\I_{\Sigma
})\longrightarrow Q_{\Sigma }\longrightarrow 0\,.
$$

When the ideal $\I_\Sigma$ is given by constraint functions as in the Dirac
approach, $W(\I_\Sigma)$ are first class functions and $\I^{\prime}_\Sigma$
are the first class constraints.

\begin{example}

We give here an example of an iterated reduction. We consider a
parametrization of $T\R^4$ in terms of the identity matrix in dimension 2
($\sigma_0$) and the $2\times 2$ Pauli matrices as follows: $\pi=p_0\sigma_0+p_a\sigma_a$ and
$g=y_0\sigma_0+y_a\sigma_a$.  

A preliminary ``constraint'' manifold is selected by requiring that 
$$
\Tr g^+g=1 \quad \Tr g\pi^+=0.
$$

This manifold is diffeomorphic to the tangent bundle of $S^3$,
i.e. $TS^3$. The Hamiltonian $H=(p_\mu p^\mu)(y_\mu y^\mu)$ defines a vector field
tangent to the constraint manifold. Similarly for the ``potential'' function
$V=\frac 12 (y_0^2+y_3^2-y_1^2-y_2^2)/y_\mu y^\mu$. 

The Hamiltonian function $\frac 12 H +V$, when restricted to $TS^3$ with a
slight abuse of notation  acquires
the suggestive form
$$
H=\frac 12 (p_0^2+p_3^2+y_0^2+y_3^2)+\frac 12 (p_1^2+p_2^2-y_1^2-y_2^2). 
$$

By using the relation $y_0^2+y_1^2+y_2^2+y_3^2=1$ we may also write it in the form
$$
H=\frac 12 (p_0^2+p_3^2+2(y_0^2+y_3^2))+\frac 12 (p_1^2+p_2^2)-\frac 12. 
$$

Starting now with $TS^3$ we may consider the further reduction by fixing
$$
\Sigma_K=\{ (y_\mu, p^\nu)\in TS^3\mid  y_0p_3-p_0y_3+y_1p_2-y_2p_1=K\}, 
$$
and quotienting by the vector field
$$
X=y_0\pd{y_3}-y_3\pd{y_0}+y_1\pd{y_2}-y_2\pd{y_1}+
p_0\pd{p_3}-p_3\pd{p_0}+p_1\pd{p_2}-p_2\pd{p_1}.  
$$

The final reduced manifold will be $TS^2\subset T\R^3$, with projection $TS^3\to TS^2$
provided by the tangent of the Hopf fibration $\pi:S^3\to S^2$, defined as
$$
x_1=2(y_1y_3-y_0y_2) \quad x_2=2(y_2y_3-y_0y_1)\quad
x_3=y_0^2+y_3^2-y_1^2-y_2^2 .
$$

The final reduced dynamics will be associated with the Hamiltonian function of
the spherical pendulum.

The spherical pendulum is thus identified by
$$
S^2\subset \R^3=\{ x\in \R^3 \mid  \langle x, x\rangle=x_1^2+x_2^2+x_3^2=1\} 
$$
$$
TS^2\subset T\R^3=\{ (x,v)\in \R^3\times \R^3 \mid  \langle x, x\rangle =1, \,\,
\langle x, v\rangle =0\} 
$$

The dynamics is given by means of $\omega=\sum_idx_i\land dv_i$ when
restricted to $TS^2$, in terms of $E=\frac 12 \langle v,v\rangle +x_3$. The
angular momentum is a constant of the motion corresponding to the rotation
around the $Ox_3$ axis. The energy momentum map
$$
\mu:TS^2\to \R^2: (x,v)\mapsto (E(x,v), L(x,v))
$$
has quite  interesting properties as shown by \cite{CushBat:1997,Duis:1980}. 

\end{example}

\section{Formulations of Quantum Mechanics}

Having defined Poisson reduction we are now on the good track to define a
reduction procedure for quantum systems. After all, according to deformation
quantization the Poisson bracket provides us with a first order
approximation to Quantum Mechanics. However, before entering a general
discussion of the reduction procedure for quantum systems, let us recall
very briefly the various formalisms to describe quantum dynamical evolution.

The description of quantum systems is done basically by means of either the Hilbert
space of states, where we define dynamics by means of the Schr\"{o}dinger
equation, or by means of the algebra of observables, where dynamics is
defined by means of the Heisenberg equation. We may also consider other
pictures like the Ehrenfest picture, the phase-space picture and the $\C%
^{\ast }$--algebra approach.

\subsection{The Schr\"odinger equation in Wave Mechanics}
\subsubsection{The framework}
Let us consider first the usual description of Schr\"odinger formulation of
Quantum Mechanics. We consider the set of states of our quantum system  to be
the space of square integrable functions on some domain $D$ (which, for
simplicity can be assumed to be some open subset of $\R^n$ but that can also be
considered to be a general differential manifold, possibly with boundary). Thus
the Hilbert 
space describing the set of states will be $L^2(D, d\mu)$, where we denote by
$d\mu$ the measure associated to a volume form. The states themselves will be denoted as
$\psi$ or as $|\psi \rangle$, in the standard bra-ket notation. Observables are required to be
symmetric operators on this space, and usually realized as differential
operators. In this setting, dynamics is introduced through the Schr\"odinger equation: 
\begin{equation}
i\hbar \frac{d}{dt}\psi =H\psi \quad \psi \in \Hil  \label{eq:schrodinger}
\end{equation}
where the Hamiltonian operator $H$ corresponds to an essentially  self-adjoint
differential operator acting on $L^2(D, d\mu)$ and written as
\begin{equation}
  \label{eq:Hamilt}
  H\psi=\left ( -{\hbar^2}\frac{\partial^2}{\partial x^2}+V(x)\right ) \psi(x),
\end{equation}
where $V(x)$ represents the potential energy of the system usually
assumed to act as a multiplicative operator.

Dynamics can also be encoded in a unitary operator $U(t,t_0)$ such that
$$
|\psi(t)\rangle =U(t,t_0)|\psi(t_0)\rangle. 
$$
It is possible to write a differential equation to encode Schr\"odinger equation
on $U$ by setting
$$
i\hbar \frac{d}{dt}U(t,t_0)=HU(t,t_0).
$$

By using eigenstates of the position operator $Q| x \rangle=x|x\rangle $, we can write the
identity in the form
$$
\mathbb{I}=\int_D| x\rangle  dx \langle x |.
$$
Then we can give the operator $U$ an integral form:
$$
\psi(x,t)=\langle x, \psi(t)\rangle =\int_D G(x, t;x_0,t_0)\psi(x_0, t_0)dx_0
$$
where $G(x,t;x_0,t_0)=\langle x, U(t,t_0)x_0\rangle $. The function $G$ is known as the
propagator or the Green function of the system.

Schr\"odinger equation exhibits interesting properties but we would like to focus now on
the fact that it can be given a Hamiltonian form with respect to the symplectic
structure that can be associated to the imaginary part of the Hermitian
structure of the Hilbert space, when considered as a real manifold.
We will elaborate a little further on this statement in the next sections.

Keeping this in mind we can think now of the analogue of the reduction procedures
that we have seen in the classical setting. The idea is quite the same for
simple examples such as the free motion, but the quantum nature of the system
provides us with some new features:

\subsubsection{Example: The reduction of free motion in the quantum case}

The description of the free quantum evolution is rather simple because the
semi-classical treatment is actually exact  \cite{EMS:2004}. In what follows we are
setting $\hbar=1$ for simplicity.

The Hamiltonian operator for free motions in two dimensions, written in polar
coordinates is
$$
H=-\frac 12 \frac 1Q \pd{Q} Q \pd{Q} -\frac 1{Q^2}\pdd{}{\phi}\,.
$$

By a similarity transformation $H'=Q^{\frac 12}HQ^{-\frac 12}$ we get rid of
the linear term and obtain
$$
H'=-\frac 12 \left ( \pdd{}{Q}+\frac 1{Q^2}\left ( \frac 14 +\pdd{}{\phi} \right )
\right )\,.
$$

Restricting $H'$ to the subspace of square
integrable functions of the form $\S_m=\{ \psi=e^{im\phi}f(Q)\} $, we find that on
this particular subspace
$$
H'\psi=-\frac 12 \left ( \pdd{}{Q}-\frac 1{Q^2}\left (m^2-\frac 14\right )\right)\psi.
$$

Going back to the Hamiltonian operator, we find  
$$
Q^{\frac 12}HQ^{-\frac 12}(Q^{\frac 12}\psi)=-\frac 12 \left ( \pdd{}{Q}-\frac
  1{Q^2}\left (m^2-\frac 14\right )\right )Q^{\frac 12}\psi =EQ^{\frac 12 }\psi. 
$$
This determines a Hamiltonian operator along the radial coordinate and setting
$g^2=m^2-\frac 14$ we have
$$
\widetilde H=-\frac 12 \pdd{}{Q}+\frac 12 \frac {g^2}{Q^2}.
$$
If we parametrize the Euclidean space with matrices $X$,
solutions of the free problem, in generic coordinates $\{ X\}$, are given, of
course, by wave-packets formed out of ``plane-waves''
$$
\psi_P(X)=Ae^{i\Tr XP},
$$
where $A$ is a normalization constant, chosen in such a way as to give a delta
function normalization.

By decomposing $X$ into a ``radial'' part $Q$ and an ``angular'' part $G$, say
$X=G^{-1}QG$,  we can write the wave function in the form

$$
\psi(Q, G)=Ae^{i\Tr (G^{-1}QGP)}=\psi_P(X).
$$

In this particular case it is not   difficult to show that $I_j(X,P)=\Tr(P^j)$
are constants of the motion in involution and give rise to the operators
$(-i)^j\Tr \left (\pd{X}\right )^j$.

To perform specific computations let us go back to the two-dimensional
situation. We consider $\psi_P=A e^{i\Tr PX}$  and project it along the
eigenspace $\S_m$ of the angular momentum corresponding to
the fixed value $m$. 

We recall that (in connection with the unitary representations of the Euclidean
group)
$$
\int_0^{2\pi}d\phi\, e^{im\phi}\,e^{iPQ\cos \phi}=2\pi J_m(PQ),
$$
where $J_m$ is the Bessel function of order $m$. Thus we conclude 
$$
\psi_P(Q)=2\pi \sqrt{PQ}\,J_m(PQ).
$$

In the particular case we are considering  free motion is described by a
quadratic Hamiltonian in $\R^2$. Therefore the Green function becomes
$$
G(X_t-X_0, 0; t)=\frac C{2t}e^{i\frac{\Tr (X_t-X_0)^2}{t}}\,.
$$
The Green function can be written in terms of the action (the solution of the
Hamilton-Jacobi equation, see Section \ref{sec:lagr-descr-solut}) and the Van
Vleck determinant (\cite{EMS:2004}, appendix 4.B).  

By using polar coordinates the kernel of the propagator is
\begin{eqnarray}
  G(Q_t, Q;t)=\sqrt{Q_tQ_0}\int_0^{2\pi}d\phi e^{im\phi}K(X_t, X_0;t)=
\nonumber \\
=\sqrt{Q_tQ_0}\frac{e^{i\frac{(Q_t^2+Q_0)^2}{2t}}}{2\pi
  it}\int_0^{2\pi}d\phi e^{im\phi}e^{-i\frac {Q_tQ_0\cos \phi}t} =
\sqrt{Q_tQ_0}\frac{e^{i\frac{(Q_t^2+Q_0)^2}{2t}}}{2\pi it} J_m\left (
  \frac{Q_tQ_0}t \right ),
\end{eqnarray}
where the angle is coming from the scalar product of $X_t$ with $X_0$.

\subsection{Reduction in terms of Differential operators}

With this simple example we have discovered that in wave mechanics the
reduction procedure involves differential operators and their eigenspaces. Let
us therefore consider some general aspects of reduction procedures for
differential operators.

In general, the Hamiltonian operator defining the Schr\"odinger equation on
$L^2(\D, d\mu)$ is a differential operator, which may exhibit a complicated
dependence in the potential. It makes sense thus to study a general framework
for the reduction of differential operators acting on some domain $D$, when we
assume that the reduction procedure consists in the suitable choice of some
``quotient'' domain $\D'$.

\subsubsection{Abstract definition of Differential operators}

We consider $\F=C^\infty (\R^n)$, the algebra of smooth functions on $\R^n$.
A differential operator of degree at most $k$ is defined as a linear map $%
D^k:\F\to \F$ of the form 
\begin{equation}  \label{eq:diff}
D^k=\sum_{|\sigma|\leq k}g_\sigma \frac{\partial^{|\sigma|}}{\partial
x_\sigma}, \quad g_\sigma\in \F
\end{equation}
where $\sigma=(i_1, \cdots i_n)$, $|\sigma| =\sum_ki_k$ and 
$$
\frac{\partial^{|\sigma|}}{\partial x_\sigma}= \frac{\partial^{|\sigma|}}{%
\partial x_1^{i_1}\cdots \partial x_n^{i_n}} 
$$

This particular way of expressing differential operators relies on the
generator of ``translations'', $\pd{x_k}$. Therefore, when the reduced space
does not carry an action of the translation group this way of writing
differential operators is not very convenient.
There is an intrinsic way to define differential operators which does not
depend on coordinates \cite{AVL:1991,GrabPon:2004b,GrabPon:2004}. One starts
from the following observation 
$$
\left [ \frac \partial{\partial x_j}, \hat f \right ]=\widehat{\frac{\partial f}{%
\partial x_j}}, 
$$
where $\hat f$ is the multiplication operation by $f$, i.e. an operation of
degree zero $\hat f:g\mapsto fg$, with $f,g\in \F$.

It follows that 
$$
[D^k, \hat f]=\sum_{|\sigma|\leq k}g_\sigma \left [\frac{\partial^{|\sigma|}%
}{\partial x_\sigma}, \hat f\right ], 
$$
is of degree at most $k-1$. Iterating for a set of $k+1$ functions $f_0,
\cdots ,f_k\in \F$, one finds that 
$$
[\cdots,[D^k, \hat f_0],\hat f_1], \cdots, \hat f_k]=0; 
$$

This algebraic characterization allows for a definition of differential
operators on any manifold.

The algebra of differential operators of degree 1 is a Lie subalgebra with respect
to the commutator and splits  into a direct sum
$$
D^1=\F\oplus D_c^1 
$$
where $D_c^1$ are derivations, i.e. differential operators of degree one which give
zero on constants. We can endow the set with a Lie algebra structure by setting
$$
[(f_1, X_1), (f_2,X_2)]=(X_1f_2-X_2f_1, [X_1, X_2])
$$

If we consider $\F$ as an Abelian Lie algebra, $D_c^1$ is
the algebra of its derivations and then $D^1$ becomes what is known in the literature as
the ``holomorph'' of $\F$  \cite{CarIbNArPere:1994}. In
this way the algebra of differential operators becomes the enveloping
algebra of the holomorph of $\F$.

The set of differential operators on $M$, denoted as $\mathcal{D}(M)$, can
be given the structure of a graded associative algebra and it is also a
module over $\F$. Notice that this property would not make sense at the level
of abstract operator algebra. To consider the problem of
reduction of differential operators we consider the problem of reduction
of first order differential operators. Because the zeroth order ones are just
functions, we restrict our attention to vector fields, i.e. the set $D_c^1$.

Given a projection $\pi:M\to N$ between smooth manifolds, we say that a
vector field $X_M$ projects onto a vector field $X_N$ if 
$$
L_{X_M}\pi^*f=\pi^*(L_{X_N}f) \qquad \forall f\in \F(N). 
$$
We say thus that $X_M$ and $X_N$ are $\pi$--related.

Thus if we consider the subalgebra $\pi^*(\F(N))\subset \F(M)$, a vector
field is projectable if it defines a derivation of the subalgebra $\pi^*(\F%
(N))$. More generally,  for a differential operator $D^k$, we shall say that it is
projectable if 
$$
D^k\pi^*(\F(N))\subset \pi^*(\F(N)). 
$$

It follows that projectable differential operators of degree zero are
elements in $\pi^*(\F(N))$. Therefore projectable differential operators are
given by the enveloping algebra of the holomorph of $\pi^*(\F(N))$, when the
corresponding  derivations are considered as  belonging to $\mathfrak{X}(M)$.

\begin{remark}
Given a subalgebra of differential operators in $\mathcal{D}(M)$ it is not
said that it is the enveloping algebra of the first order differential
operators it contains. When this happens, we cannot associate a corresponding quotient
manifold with an identified subalgebra of differential operators. An example of
this situation arises with angular momentum operators when we consider the
``eigenvalue problem'' in terms of $J_z$ and $J^2$. It is clear that this
commuting subalgebra of differential operators can not be generated by its
``first order content''. 

In the quantization procedure, this situation gives rise to anomalies
\cite{AGM:1998}. 
\end{remark}

\subsubsection{Example: differential operators and the Kustainheimo-Stiefel (KS) fibration}

In this section we would like to consider the reduction of differential
operators associated with the KS projection $\pi_{KS}:\R_0^4\to \R_0^3$, where 
$\R^j_0=\R^j-\{0\}$, and show that the hydrogen atom operator may be obtained as
a reduction of the operators associated with a family of harmonic oscillators.

Let us recall first how this map is defined. We first notice that $\R^4_0=S^3\times
\R^+\sim SU(2)\times \R^+$ and $\R^3_0=S^2\times \R^+$. By introducing polar coordinates 
$$
g=Rs \quad s\in SU(2), \quad R\in \R^+, 
$$
we define $\pi_{KS}:\R^4_0\to \R^3_0$ as 
$$
\pi_{KS}: g\mapsto g\sigma_3g^{+}=R^2s\sigma_3 s^{-1}=x^k\sigma_k, 
$$
where $\{ \sigma_k\} $ are the Pauli matrices. In a Cartesian coordinate
system one has 
$$
x_1=2(y_1y_3+y_2y_0) \quad x_2=2(y_2y_3-y_1y_0) \quad
x_3=y_1^2+y_2^2-y_3^3-y_0^2,
$$
where $g=\sum_iy_i\sigma^i$. Moreover, $\sqrt{x^jx_j}=r=R^2=y^ky_k$.

The KS projection defines a principal fibration with structure group $U(1)$.

By the definition of $\pi_{KS}$ it is easy to see that acting with  $%
e^{i\lambda \sigma_3}$ on $SU(2)$ does not change the projected point on $\R%
^3_0$. The associated fundamental vector field is the left
invariant infinitesimal generator associated with  $\sigma_3$, i.e. $i_{X_3}s^{-1}
ds=i\sigma_3$. In coordinates it reads 
$$
X_3=y^3\frac\partial {\partial y^3}-y^3\frac\partial {\partial
y^0}+y^1\frac\partial {\partial y^2}-y^2\frac\partial {\partial y^1} 
$$

We consider the Lie algebra of differential operators generated by $X_3$ and $%
\pi^*_{KS} (\F(\R^3_0))$. Projectable differential operators with respect to
$\pi_{KS}$ are given by 
the normalizer of this algebra in the algebra of differential operators $%
\mathcal{D}(\R^4_0)$. As we already remarked this means that this subalgebra must
map $\pi^*(\F(\R^3_0))$ into itself. If we denote this subalgebra by $%
\mathcal{D}^\pi$ we may also restrict our attention to the operators in $%
\mathcal{D}^\pi$ commuting with $X_3$. To explicitly construct this algebra
of differential operators we use the fact that $SU(2)\times \R_+$ is a Lie
group and therefore it is parallelizable. Because the KS map has been
constructed with left invariant vector field  $X_3$, we consider the generators of the
left action of $SU(2)$, say right invariant vector fields $Y_1, Y_2, Y_3$,
and a central vector field along the radial coordinate, say $\mathcal{R}$.
All these vector fields are projectable and therefore along with $\pi_{KS}^*(%
\F(\R^3_0)$ generate a projectable subalgebra of differential operators
which covers the algebra of differential operators on $\R^3_0$. 
This map is surjective and  we can ask to find the ``inverse image'' of the operator $%
\hat H=-\frac{\Delta_3}2-\frac kr$, which is the operator associated with
the Schr\"odinger equation of the hydrogen atom ($\Delta_3$ denotes the
Laplacian in the three dimensional space). As this operator is invariant
under the action of $\mathfrak{so}(4)\sim \mathfrak{su}(2)\oplus
\mathfrak{su}(2)$, associated with the angular 
momentum and the Runge-Lenz vector, we may look for a representative in the
inverse image which shares the same symmetries. As the pull-back of the
potential $\frac kr$ creates no problems, we may concentrate our attention
on the Laplacian. Because of the invariance requirements, our candidate for
the inverse image will have the form 
$$
D=f(R)\frac{\partial^2}{\partial R^2}+g(R)\frac \partial{\partial R}%
+h(R)\Delta^s_3+c(R), 
$$
where $R$ is the radial coordinate in $\R^4_0$, and $f,g,h$ are functions to be
determined. We recall that in polar coordinates the Laplacian $\Delta_3$ has the
expression 
$$
\Delta_3=\frac{\partial^2}{\partial r^2}+\frac 2r \frac \partial{\partial r}%
+\frac 1{r^2}\Delta^s_2, 
$$
where we denote by $\Delta^s_n$ the Laplacian on the $n$--dimensional sphere.

By imposing $D\pi^*_{KS}f=\pi^*_{KS}(\hat H_3 f)$ for any $f\in \F(\R^3_0)$
we find that the representative in the inverse image has the expression 
$$
H^{\prime}=-\frac 12 \frac 1{4R^2}\Delta_4-\frac k{R^2}. 
$$
This operator is usually referred to as the conformal Kepler Hamiltonian
\cite{DavanMar:2005}. 

Now, with this operator we may try to solve the eigenvalue problem 
$$
\left ( -\frac 12 \frac 1{4R^2}\Delta_4-\frac k{R^2}\right ) \psi -E\psi=0 
$$

It defines a subspace in $\F(\R^4_0)$ which coincides with the one
determined by the equation 
$$
\left ( -\frac 12 \Delta_4-4E{R^2}-4k \right ) \psi =0 .
$$

This implies that the subspace is given by the eigenfunctions of the
Harmonic oscillator with frequency $\omega(E)=\sqrt{-8E}$. We notice then
that a family of oscillators is required to solve the eigenvalue problem
associated with the hydrogen atom. To find the final wave functions on $\R^3$
we must require that 
$L_{X_3}\psi=0$ in
order to find eigenfunctions for the three dimensional problem. Eventually
we find the correct relations for the corresponding eigenvalues 
$$
E_,=-\frac{k^2}{2(m+1)^2}, \quad m\in \mathbb{N}. 
$$

Of course, dealing with Quantum Mechanics we should ensure that the operator 
$H^{\prime}=-\frac 12 \frac1{4R^2}\Delta_4-\frac k{R^2}$ is essentially
self-adjoint to be able to associate with it a unitary dynamics. One finds
that the Hilbert space should be constructed as a space of square integrable
functions on $\R^4_0$ with respect to the measure $4R^2d^4y$ instead of the
Euclidean measure on $\R^4$. We shall not go into the details of this, but
the problem of a different scalar product is strictly related to the
reparametrization of the classical vector field, required to turn it into a
complete one. This would be a good example for J. Klauder's saying: ``
these are classical symptoms of a quantum illness'' (see
\cite{ZhyKlau:1993}). Further details can be found in
\cite{DavanMar:2005,DavanMarVal:2005}.  As for the reduction of the Laplacian
in Quantum Mechanics see also \cite{FePus:2006,FePus:2007b,FePus:2007}. For a
geometrical approach to the problem of self-adjoint extensions see
\cite{AsIbMar:2005}.

\subsection{Heisenberg formalism}

A different approach to Quantum Mechanics is given by what is known as the Heisenberg
picture. Here dynamics is encoded in the algebra of observables,
considered as the real part of an abstract $\C^*$--algebra.  

First, we have to consider  observables as associated with Hermitian operators (finite
dimensional matrices if the system is finite dimensional). These
matrices do not define an associative algebra because the product of two
Hermitian matrices is not Hermitian. However we may complexify this space by
writing a generic matrix as the sum of a real part $A$ and an imaginary part
$iB$,  $A$ and $B$ being Hermitian. In this way we find that:

\begin{proposition}
The complexification of the algebra of observables allows us to write an
associative product of operators $A=A_{1}+iA_{2}$, where $A_{1}$ and $A_{2}$
are real  Hermitian. We shall denote by $\mathcal{A}$ the corresponding
associative algebra. 
\end{proposition}

Finally we can proceed to define the equations of motion  on this
complexified algebra of observables. It is introduced by means of the 
\textbf{Heisenberg equation}: 
\begin{equation}
i\hbar \frac{d}{dt}A=[A,H]\,,\quad A\in \mathcal{A},  \label{eq:heisenberg}
\end{equation}
where $H$ is called the Hamiltonian of the system we are describing. To take
into account an explicit time-dependence of the observable we may also write
the equation of motion in the form
\begin{equation}
\frac{d}{dt}A=-\frac{i}{h}[A,H]+ \frac{\partial A}{\partial t}\quad A\in 
\mathcal{A}.
\end{equation}

From a formal point of view, this expression is similar to the expression of
Hamilton equation written on the Poisson algebra of classical observables
(i.e. on the algebra of functions representing the classical quantities 
with the structure provided by the Poisson bracket we assume our classical
manifold is endowed with).  This similarity is not casual and turns out to be
very useful in the study of the quantum-classical transition. We shall come
back to this point later on.

\begin{remark}
The equations of motion written in this form are necessarily derivations of
the associative product and can therefore be considered  as ``intrinsically
Hamiltonian''. In the Schr\"{o}dinger picture, however, if the vector field is not
anti-Hermitian, the equation still makes sense, but the dynamics need not be 
K\"ahlerian. To recover a similar treatement, one has to give up the requirement
that the evolution preserves the product structure on the space of
observables. 
\end{remark}

This approach to Quantum Mechanics relies on the non-commutative algebra of
observables, therefore it is instructive to consider a reduction procedure for
non-commutative algebras.

\subsubsection{ An example of reduction in a non-commutative setting}
The  example of reduction procedure in a non-commutative setting that we are
going to discuss
reproduces the Poisson reduction in the ``quantum-classical'' transition and goes
back to the celebrated example of the quantum $SU(2)$  written by Woronowicz
\cite{Woro:1987} and is adapted  from \cite{GrabLanMarVi:1994}.  

We consider the space $S^3\subset \R^4$, identified with the group $SU(2)$
represented in terms of matrices.  The $\star$--algebra $\A$ generated by
matrix elements
is dense in the algebra of  continuous functions on $SU(2)$  and can be
characterized as the ``maximal'' unital commutative $\star$--algebra $\A$,
generated by elements which we can denote as $\alpha, \nu, \alpha^*, \nu^*$ satisfying
$\alpha^* \alpha+\nu^*\nu =1$. This algebra can be generalized and deformed into a
non-commutative one by replacing some relations with the following ones: 
$$
\alpha \alpha^*-\alpha^* \alpha =(2q-q^2)\nu ^* \nu \qquad
\nu ^* \nu - \nu \nu^*=0 
$$
and 
$$
\nu \alpha -\alpha \nu =q\nu \alpha \qquad \nu ^*\alpha -\alpha \nu ^*=q\nu ^* \alpha.
$$

This algebra reduces to the previous commutative one  when $q=0$. In this respect this
situation resembles the one on the phase-space where we consider ``deformation
quantization'' and the role of the parameter $q$ is
played by the Planck constant.  Pursuing this analogy we may consider the formal
product depending on the parameter $q$:
$$
u\star_q v= uv+\sum_n q^n P_n(u,v),
$$
where $P_n$ are such that the product $\star_q$ is associative. 

Since the commutator bracket 
$$
[u,v]_q=u\star_qv-v\star_qu
$$
is a biderivation (as for any associative algebra)  and satisfies the Jacobi
identity  we find that the ``quantum Poisson bracket''  gives a Poisson bracket
when restricted to ``first order elements''
$$
\{ u,v\} =P_1(u,v)-P_1(v,u).
$$
In general, we can write
\begin{equation}
  \label{eq:limit}
\lim_{q\to 0}\frac 1 q[u,v]_q=\{ u, v\}   
\end{equation}

From the defining commutation relations written by Woronowicz we get the
corresponding quadratic Poisson brackets on the matrix elements of $SU(2)$:
$$
\{ \alpha , \bar \alpha\} =2\bar \nu \nu , \quad \{ \nu, \bar \nu\} =0, \quad \{ \nu, \alpha\} =\nu
\alpha, \quad \{ \bar \nu , \alpha \} =\bar \nu \alpha .
$$

Passing to real coordinates, $\alpha=q_2+ip_2$ and $\nu=q_1+ip_1$, we get a purely
imaginary bracket whose imaginary part is the following quadratic Poisson bracket 

\begin{eqnarray*}
\{ p_1, q_1\} &=&0, \,\,\, \, \{ p_1, p_2\} =q_1q_2, \,\,\,\, \{ p_1, q_1\} =-p_1p_2, \\
\{ q_1, p_2\} &=&q_1q_2, \,\,\, \,  \{ q_1, q_2\} =-q_1p_2, \,\,\,\,  \{ p_2, q_2\}=q_1^2+p_1^2. 
\end{eqnarray*}

The functions $q_1^2+q_2^2+p_1^2+p_2^2$ is  a Casimir function for this
Lie algebra.

By performing a standard Poisson bracket reduction  we find a bracket on
$S^3$. If we identify this space with the group $SU(2)$ we get the Lie-Poisson
structure on $SU(2)$: 

The vector field 
$$X=-q_1\pd{p_1}+p_1\pd{q_1}+q_2\pd{p_2}-p_2\pd{q_2}$$ selects a
subalgebra of functions $\F$ by imposing  the condition
$L_X\F=0$. This reduced algebra can be regarded as the algebra generated by 
$$
u=-p_1^2-q_1^2+p_2^2+q_2^2, \qquad \nu=2(p_1p_2+q_1q_2), \qquad
z=2(p_1q_2-q_1p_2),
$$
with brackets
$$
\{ v,u\} =2(1-u)z, \qquad \{ u,z\} =2(1-u)v, \qquad \{ z, v\} =2(1-u)u.
$$

One finds that $u^2+v^2+z^2=1$ so that the reduced space of $SU(2)$ is the
unit sphere $S^2$ and the reduced bracket vanishes  at the North Pole ($u=1,
v=z=0$). 

It may be interesting to notice that the stereographic projection from the
North Pole pulls-back the standard symplectic structure on $\R^2$ onto the one
associated with this one on $S^2-\{ \mathrm{ North \, \,  Pole }\}$.

It is now possible to carry on the reduction at the non-commutative level. We
identify the subalgebra $\A'_q\subset \A_q$ generated by the elements 
$u=\mathbb{I}-2\nu^*\nu =\alpha^*\alpha-\nu^*\nu$, $w=2\nu^*\alpha$ and $w^*=2\alpha^*\nu
$. We have $uu^*+w^*w=\mathbb{I}$ and the algebra $\A'_q$  admits a limit given
by $\A'_0$ generated by the two dimensional sphere $S^2$. The subalgebra
$\A'_q$ can be considered as a quantum sphere.

The quantum Poisson bracket on  $S^2$ is given by
$$
[w,u]=(q^2-2q)(1-u)w, \qquad [w^*, u]=-(q^2-2q)(1-u)w^*,
$$
and 
$$
[w,w^*]=-(2q^2-2q)(1-u)+(4q-6q^2+4q^3-q^4)(1-u)^2.
$$

Passing to the classical limit we find, by setting $v=\mathrm{Re}(w)$,
$z=-\mathrm{Im}(w)$: 
$$
\{ v,u\}=2(1-u)z, \qquad \{ u,z\} =2(1-u)u, \qquad \{ z,v\} =2(1-u)u,
$$
which coincides with the previous reduced Poisson bracket associated with the
vector field $X$. In this case, the reduction procedure commutes with the
``quantum-classical'' limit.

In this same setting it is now possible to consider a ``quantum dynamics'' and
the corresponding ``classical'' one to see how they behave with respect to the
reduction procedure.

On the algebra $\A_q$ we consider the dynamics defined by the Hamiltonian 
$$
H=\frac 12 u=\frac 12 (\mathbb{I}-2\nu^*\nu)=\frac 12 (\alpha^*\alpha-\nu^*\nu).
$$
This choice ensures that our Hamiltonian defines a dynamics on $\A'_q$. The
resulting equations of motion are
$$
[H, \nu]=0, \quad [H, \nu^* ]=0, \quad [H,\alpha]=(q^2-2q)\nu ^*\nu \alpha, \quad 
[H, \alpha^*]=-(q^2-2q)\nu^* \nu \alpha^*,
$$
so that the dynamics written in the exponential form is
$$
U(t)=e^{it\mathrm{ad}_H}
$$
and gives,
$$
\nu(t)=\nu_0 , \qquad \nu^*(t)=\nu^*(0)
$$
$$
\alpha(t)=e^{it(q^2-2q)\nu^*\nu}\alpha_0, \qquad
\alpha^*(t)=e^{-it(q^2-2q)\nu^*\nu}\alpha_0^*.
$$

Going to the ``classical limit'' we  find 
$$
H=\frac 12 ( q_2^2+p_2^2-q_1^2-p_1^2),
$$
with the associated vector field on $S^3$ given \cite{LiMarSpaVit:1993} by 
$$
\Gamma=2(q_1^2+p_1^2)\left ( q_2\pd{p_2}-p_2\pd{q_2} \right ),
$$
the corresponding solutions are given by
$$
q_1(t)=q_1(0), \qquad p_1(t)=p_1(0)
$$
$$
p_2(t)=\cos (2t(q_1^2+p_1^2))p_2(0)+\sin(2t(q_1^2+p_1^2))q_2(0),
$$
$$
q_2(t)=-\sin (2t(q_1^2+p_1^2))p_2(0)+\cos(2t(q_1^2+p_1^2))q_2(0).
$$

If we remember (\ref{eq:limit}),
this flow is actually the limit of the quantum flow when we take the limit of
te deformation parameter $q\to
0$ and hence $q^2/q\to 0$, $q/q\to 1$. Indeed in this case $\nu^*\nu=q_1^2+p_1^2$ and
$\alpha=q_2+ip_2$. As the Hamiltonian was chosen to be an element of $\A'_q$ we get
a reduced dynamics given by
$$
[H, w]=-\frac 12 (q^2-2q)(1-u)w, \qquad
[H, w^*]=-\frac 12 (q^2-2q)(1-u)w^*.
$$

The corresponding solutions for the endomorphism $e^{it\mathrm{ad}_H}$ become
$$
w(t)=e^{-it \frac 12 (q^2-2q)(1-u)}w(0), \qquad
w^*(t)=e^{it \frac 12 (q^2-2q)(1-u)}w^*(0).
$$

Passing to the classical limit we find the corresponding vector field on $\R^3$
tangent to $S^2$
$$
\widetilde \Gamma =(1-u)\left ( z\pd{v}-v\pd{z} \right ),
$$
which is the reduced dynamics
\begin{eqnarray}
\frac {du}{dt}&=&0\cr \quad \frac{dv}{dt}&=&2(q_1^2+p_1^2)(q_2p_1-p_2q_1)=(1-u)z,\cr
\frac{dz}{dt}&=&-2(q_1^2+p_1^2)(p_1p_2+q_1q_2)=-(1-u)v.
\end{eqnarray}

By using the stereographic projection $S^2\to \R^2$ given by $(x,y)=\frac
1{1-u}(v,z)$ we find the associated vector field on $\R^2$
$$
\Gamma(x,y)=\frac 2{x^2+y^2+1}\left ( x\pd{y}-y\pd{x} \right ).
$$

This example is very instructive because provides us with an example of reduced
quantum dynamics that goes onto the corresponding reduced classical
dynamics, i.e. reduction ``commutes'' with ``dequantization''. Further details
can be found in \cite{GrabLanMarVi:1994}. 

\subsubsection{Example: deformed oscillators}

Another instance of a non-commutative algebra reduction is provided by the case
of the deformed harmonic oscillator.
Let us start thus by analyzing the case of deformed harmonic oscillators
described in the Heisenberg picture. By including the deformation parameter in
the picture we can deal with several situations at the same time, as we are
going to see.

We consider a complex vector space $V$ generated by $a, a^+$. Out of $V$ we
construct the associative tensorial algebra $\A=\C\oplus V\oplus (V\otimes V)\oplus (V\otimes V\otimes V)\oplus \cdots  $. A
dynamics on $V$, say
$$
\frac d {dt}a=-i\omega a, \quad \frac d{dt}a^+=i\omega a^+
$$
defines a dynamics on  $\A$  by extending it by using the Leibniz rule with
respect to the tensor product.

A  bilateral ideal $I_{r,q}$ of $\A$, generated by the relation $a^+a-qaa^++r=0$
, i.e. the most general element of $I_{r,q}$ has the form $A(a^+a-qaa^++r)B$,
with $A, B\in \A$, is also invariant under the previously defined equations of
motion. It follows then that the dynamics defines a derivation, a ``reduced
dynamics''  on the quotient algebra $\A_{r,q}=\A/I_{r,q}$. When $q=1$ and $r=0$
the dynamics becomes  a dynamics on a commuting algebra and therefore can be
considered to be a classical dynamics.  When $q=1$ and $r=\hbar$  we get back the
standard quantum dynamics of the harmonic oscillator. If we consider $r$ to be
a function of the ``number operator'' defined as $n=a^+a$ we obtain many of the
proposed deformations of the harmonic oscillator existing in the literature. In
particular, these 
deformations have been applied to the description of the magnetic dipole
\cite{LopManMar:1997}. It is clear  now that this reduction procedure may be
carried over to any realization  or 
representation of the abstract algebra and the corresponding ideal
$I_{r,q}$. In this example it is important that the starting dynamics is
linear. The extension to the universal tensorial algebra gives a kind of
abstract universal harmonic oscillator. The bilateral ideal we choose to
quotient the tensor algebra is responsible for the physical identification of
variables and may arise from a specific realization  of the tensor algebra in
terms of functions or operators.

\subsection{Ehrenfest formalism}

\subsubsection{The formalism}
This picture of Quantum Mechanics is not widely known but it arises in connection 
with the so called Ehrenfest theorem  which may be seen from the point of view
of $\star$--products  on phase space (see \cite{EMS:2004}). Some aspects of this picture have
been considered by Weinberg \cite{Wein:1989}  and more generally appear in the
geometrical formulation of Quantum Mechanics
\cite{CirManPizzo:1990I,CirManPizzo:1990II,CirManPizzo:1991, CirManPizzo:1994}.

We saw above how  Schr\"odinger picture assumes  as a starting point the Hilbert space of
states and derive the observable as real operators acting on this space of
states. The Heisenberg picture starts from the observables, enlarged by means
of complexification into a $\C^*$--algebra and derives the states as  positive
normalized linear functionals on the algebra of observables. In  the Ehrenfest
picture both spaces are considered jointly to define quadratic functions as
\begin{equation}
  \label{eq:func}
  f_A(\psi)=\frac 12 \langle \psi, A\psi\rangle. 
\end{equation}
In this way all operators are transformed into quadratic functions which are
real valued when the operators are Hermitian. The main advantage of this
picture relies on the fact that we can define a Poisson bracket on the space of
quadratic functions by setting
\begin{equation}
  \label{eq:poissonbr}
  \{ f_A, f_B\} :=if_{[A,B]},
\end{equation}
where $[A,B]$ stands for the commutator on the space of operators. By
introducing an orthonormal basis in $\Hil$, say $\{  \psi_k \} $, we may write
the function $f_A$ as
$$
f_A(\psi)=\frac 12\sum_{jk} c_jc_k^*\langle \psi_j, A\psi_k\rangle, \qquad \psi =\sum_k c_k
\psi_k 
$$
and the Poisson bracket then becomes
$$
\{ f_A, f_B\} =i\sum_k\left (
  \pdc{f_A}{c_k}\pdc{f_B}{c_k^*}-\pdc{f_A}{c^*_k}\pdc{f_B}{c_k}  \right
). 
$$

This bracket can be used to write the equations of motion in the form
$$
i\frac{df_A}{dt}=\{ f_H, f_A\}, 
$$
where $f_H$ is the function associated to the Hamiltonian operator.

While this way of writing the dynamics is very satisfactory because allows us
to write the equations of motion in a ``classical way'', one has lost the
associative  product of operators. Indeed, the point-wise product (somehow a
natural one for the functions defined on a real differential manifold) of two
quadratic functions will not be 
quadratic but a quartic function. To recover the associative product we can,
however, get inspiration from the definition of the Poisson bracket
(\ref{eq:poissonbr}) and introduce
\begin{equation}
  \label{eq:assoc}
  (f_A\star f_B)(\psi):=f_{AB}(\psi)=\frac 12 \langle \psi, AB \psi\rangle. 
\end{equation}

By inserting a resolution of the identity $\sum_j |\psi_j\rangle \langle \psi_j|
=\mathbb{I}$ (since there is a numerable basis for $\Hil$) in
between the two operators in $AB$,  say
$$
\langle \psi, A\sum_j|\psi_j\rangle \langle \psi_j|B\psi \rangle, 
$$
and writing the expression of $\psi$ in terms of the basis elements $\psi 
=\sum_kc_k\psi_k $ we find a product
$$
 (f_A\star f_B)(\psi)=\sum_{jkl}c_j c^*_l\langle \psi_j, A\psi_k\rangle \langle \psi_k,B\psi_l\rangle ,
$$ 
which reproduces the associative product of operators but now it is not
point-wise anymore. 

As a matter of fact the Poisson bracket defines derivations for this product,
i.e.
$$
\{ f_A, f_B\star f_C\} =\{ f_A, f_B\} \star f_C+f_B\star \{ f_A, f_C\} \quad \forall f_A, f_B, f_C.
$$
Therefore it is an instance of what Dirac calls a quantum Poisson bracket
\cite{Dirac:book}. In the literature it is known as a Lie-Jordan bracket
\cite{Emch,Lands:book}. 

Using both products, the Ehrenfest picture becomes equivalent to Schr\"odinger
and Heisenberg ones.

Let us consider now how the expressions of the products are written in terms of
a different basis, namely the basis of eigenstates of the position operator $Q$
or the momentum operator $P$. We have thus two basis $\{ | q\rangle \}$ and $\{ | p\rangle \} $
satisfying $Q| q\rangle =q |q\rangle $ and $P|p\rangle =p|p\rangle $ and 
$$
\int_{-\infty}^\infty |q\rangle dq\langle q | =\mathbb{I}=\int_{-\infty}^\infty |p\rangle  dp\langle p | .
$$

Now the matrix elements $A_{kj}=\langle\psi_j, A\psi_k\rangle $ of the operators in the
definition of the $\star$ product above become 
$$
A(q', q)=\langle q', Aq\rangle \mathrm{\quad  or \quad } A(p', p)=\langle p', Ap\rangle, 
$$
and the sum is replaced by an integral:

\begin{equation}
  \label{eq:star2}
(f_A\star f_B)(\psi)=\int dq dq' dq'' c(q'') c^*(q')A(q'', q) B(q,q') .  
\end{equation}

Thus  this  is a product of functions defined on $\R^n\times \R^n$ or $(\R^n)^*\times
(\R^n)^*$ , i.e. two copies of the configuration space or two copies of the
momentum space. Following an idea of Dirac \cite{Dir:1945} one may get
functions on $\R^n\times (\R^n)^*$ by using eigenstates of the position operator on
the left and  eigenstates of the momentum operator on the right:
$$
A_l(q,p)=\langle q, Ap\rangle e^{-\frac i\hbar qp},
$$
or also interchanging the roles of position and momentum:
$$
A_l(p,q)=\langle p, Aq\rangle e^{\frac i\hbar qp}.
$$

Without elaborating much on these aspects (we refer to
\cite{ChaErMarMorMuSim:2005} for details) we
simply state that the $\star$--product we have defined, when considered on phase
space, becomes the standard Moyal product.

It is now clear that we may consider the reduction procedure in terms of
non-commutative algebras when we consider the $\star$--product. We shall give a
simple example where from a $\star$--product on $\R^4$ we get by means of a
reduction procedure a $\star$--product on the dual of the Lie algebra of
$SU(2)$. Further details connected with their use in non-commutative geometry
can be found in \cite{GraLizMarVi:2002}.

\subsubsection{Example: Star products on $\mathfrak{su}(2)$}

 We are going to show how it is possible to define star products on
spaces such as $\mathfrak{su}(2)$ by using the reduction of the Moyal star
product defined on a larger space ($\R^4$ in this case). 

Let us then consider the coordinates $\{ q_1,q_2,p_1,p_2\} $ for $\R^4$, $\{
x,y,w\}$ for $\mathfrak{su}(2)$  and  the
mapping $\pi:\R^4\to \R^3\sim \mathfrak{su}(2)$  defined as:

\begin{eqnarray*}
f_1(q_1,q_2,p_1,p_2)&=&\pi^*(x)=\frac 12 (q_1q_2+p_1p_2) \\
f_2(q_1,q_2,p_1,p_2)&=&\pi^*(y)=\frac 12 
(q_1p_2-q_2p_1) \\
 f_3(q_1,q_2,p_1,p_2)&=&\pi ^*w=\frac 14 (q_1^2+p_1^2-q_2^2-p_2^2) 
\end{eqnarray*}

It is useful to consider also the pull-back of the Casimir function of
$\mathfrak{su}(2)$, $\mathcal{C}=\frac 12 (x^2+y^2+w^2)$, which becomes  
$$
\pi^*\mathcal{C}=\frac 1{32}(p_1^2+q_1^2+p_2^2+q_2^2)^2.
$$

To define a reduced star product on $\mathfrak{su}^*(2)$  we
consider the Moyal star product on the functions of $\R^4$, and select a
$\star$--subalgebra isomorphic to the $\star$--algebra of $\mathfrak{su}^*(2)$. To
identify this subalgebra we need derivations of the $\star$--product that
annihilate the algebra we are studying. We look then for a derivation $H$ which is
a derivation of both the point-wise algebra and the $\star$--algebra, to ensure that
reduction commutes with the ``classical limit''.  The commutative
point-wise product condition will identify the quotient manifold, while the
condition on the $\star$--product identifies a star product on functions
defined on the quotient. We consider thus a vector field $H$ on $\R^4$
satisfying 
$$
L_H\pi^*x=0=L_H\pi^*y =L_H\pi^*w.
$$

This condition characterizes the point-wise subalgebra of functions of $\R^4$
which are projectable on functions of $\R^3$. Such a vector field 
can be taken to be the Hamiltonian vector field
associated to the Casimir function $\pi^*\mathcal{C}$. It is simple to see that 
the Poisson subalgebra generated by the functions $\{ \pi^*x, \pi^*y, \pi^*w,
f_H\} $ where $f_H=q_1^2+q_2^2+p_1^2+p_2^2$ is the Poisson commutant of the
function $f_H$ (see \cite{GraLizMarVi:2002}). And this set is an involutive
Moyal subalgebra when we consider  the Moyal product on them, i.e. for any
functions $F, G$
$$
\{ f_H, F\} =0=\{ f_H, G\} \Rightarrow \{ f_H, F\star G\} =0.
$$

The star product on $\mathfrak{su}(2)$ is then defined as:
$$
\pi^*(F\star_{\mathfrak{su(2)}}G)=\pi^*F\star \pi^*G.
$$

As an example we can consider the product:
\begin{eqnarray*}
x_j\star_{\mathfrak{su(2)}}f(x_i)=&\\
&\left ( x_j-\frac
  {i\theta}2\epsilon_{jlm}x_l\pd{x_m}-\frac{\theta^2}{8} \left (\left
      (1+x_k\pd{x_k}\right )\pd{x_j}-\frac 12 x_j\pd{x_k}\pd{x_k}\right )\right
)f(x_i). 
\end{eqnarray*}

The same procedure may be applied to obtain a reduced star product for all
three dimensional Lie algebras (see \cite{GraLizMarVi:2002}) and to deal with a
non-commutative differential calculus \cite{MarViZam:2006}.

\section{The complex projective space as a reduction of the Hilbert 
space} 

\subsection{Geometric Quantum Mechanics}

It is possible to show that the various pictures we have presented so far
can be given an unified treatment. To this aim it is convenient to consider
a realification of the Hilbert space and to deal with our different pictures
from a geometric perspective.

Let us start by considering again the complex Hilbert space $\Hil$ which contains the
set of states of our quantum system. Originally it is considered to be a
complex vector space, but we can also look at it as a real vector space by
considering the real and imaginary parts of the vectors $\Hil\ni |\psi\rangle
=(\psi_R, \psi_I)$. The Hermitian structure of $\Hil$ is then
encoded in two real tensors, one symmetric and one skew-symmetric; which
together with the complex structure provide us with a K\"ahler structure.
Let us first discuss this point in some detail.

We consider therefore $\Hil_\R$ with the structure of a K\"ahler manifold $(%
\Hil_\R, J,g,\omega)$, i.e. a complex structure $J:T\Hil_\R\to T\Hil_\R$, a
Riemannian metric $g$ and a symplectic form $\omega$ . First of all, we are
going to make use of the linear structure of the Hilbert space (encoded in
the dilation vector field $\Delta$) to identify the tangent vectors at any
point of $\Hil_\R$. 
In this way we can consider the Hermitian structure on $\Hil_\R$ as an
Hermitian tensor on $T\Hil_\R$. With every vector we can associate a vector
field 
$$
X_\psi:\phi \to (\phi, \psi) \,.
$$

Therefore, the Hermitian tensor, denoted in the same way as the scalar
product is 
$$
\langle X_{\psi_1}, X_{\psi_2}\rangle =\langle \psi_1, \psi_2\rangle \,.
$$

Fixing an orthonormal basis $\{ |e_k\rangle \} $ of the Hilbert space allows
us to identify this product with the canonical Hermitian product of $\C^n$: 
$$
\langle \psi_1, \psi_2\rangle =\sum_k\langle \psi_1, e_k\rangle \langle e_k,\psi_2\rangle \,.
$$
The Hilbert space becomes then identified with $\C^n$. As a result, 
the group of unitary transformations on $\Hil$ becomes identified as the
group $U(n, \C)$, its Lie algebra $\mathfrak{u}(\Hil)$ with $\mathfrak{u}(n, %
\C)$ and so on.

The choice of the basis also allows us to introduce coordinates for the
realified structure: 
$$
\langle e_k, \psi\rangle =(q_k+ip_k)(\psi), 
$$
and write the geometrical objects introduced above as: 
$$
J=\pd{{p_k}}\otimes dq_k-\pd{q_k}\otimes dp_k, \quad g=dq_k\otimes
dq_k+dp_k\otimes dp_k \quad \omega=dq_k\land dp_k \,.
$$

If we combine them in complex coordinates $z_k=q_k+ip_k$ we can write the Hermitian
structure in a simple way 
$$
h=d\bar z_k\otimes dz_k 
$$

The space of observables (i.e. of Hermitian operators acting on $\Hil$)
is identified with the dual $\mathfrak{u}^{\ast }(\Hil)$ of the real
Lie algebra $\mathfrak{u}(\Hil)$, by means of the scalar product existing on
the Lie algebra and using the fact that the multiplication of an Hermitian
matrix by the imaginary unit gives an element in the Lie algebra.  Then we
get 
$$
A(T)=\frac{i}{2}\Tr AT \qquad A\in \u^*, \quad T\in  \u \,.
$$

The product given by the trace allows us to establish an isomorphism
between $\mathfrak{u}(\Hil)$ and $\mathfrak{u}^{\ast }(\Hil)$ and identifies
the adjoint and the coadjoint action of the Lie group $U(\Hil)$.

Under the previous isomorphism, $\mathfrak{u}^{\ast }(\Hil)$ becomes a Lie
algebra with Lie bracket defined by 
$$
[A,B]_{-}=i[A,B]=i(AB-BA)\,.
$$

Moreover, we can also define a  scalar product on $\u^*$ , given by: 
$$
\langle A,B\rangle =\frac{1}{2}\Tr AB,
$$
which turns the vector space into a real Hilbert space.

The identification of vectors and covectors allows to write the isomorphism
from $\mathfrak{u}^*$ to $\mathfrak{u}$. The metric becomes then 
$$
\langle \hat A, \hat B\rangle_{\mathfrak{u}}=\frac 12 \Tr AB .
$$

We can also associate complex valued functions to linear operators $A\in %
\mathfrak{gl}(\Hil)$ as we have seen in the Ehrenfest picture,
$$
\mathfrak{gl}(\Hil)\ni A\mapsto f_A=\frac 12 \langle \psi, A\psi\rangle_{\Hil%
}. 
$$
We can use this mapping in the dual $\mathfrak{u}(\Hil)$. The way to do
it is to consider the complexification of $\mathfrak{u}(\Hil)$ and then
consider general linear transformations (i.e. elements of $\mathfrak{gl}(\Hil%
)$) and associate with them the complex valued functions we saw above.

Now, by using the contravariant form $G+i\Omega$ of the Hermitian tensor 
given by: 
$$
G+i\Omega=\frac{\partial}{\partial{q_k}}\otimes \frac{\partial}{%
\partial{q_k}}+ \frac{\partial}{\partial{p_k}}\otimes \frac{\partial}{%
\partial{p_k}}+ i\frac{\partial}{\partial{q_k}}\land \frac{\partial}{%
\partial{p_k}}, 
$$

it is possible to define a bracket (see for instance \cite{CaC-GMar:2007})
$$
\{ f,h\} _{\Hil}=\{ f,h\}_{g}+i\{ f,h\} _{\omega} 
$$

In particular, for quadratic functions we have 
$$
\{ f_A,f_B\} =f_ {AB+BA}=2f_{A\circ B} \quad \{
f_A,f_B\}_{\omega}=-if_{AB-BA} 
$$

Thus in this way we can define a tensorial version of the symmetric
product on the space of Hermitian matrices which defines a Jordan
algebra,along with the Lie product given by the commutator.

For Hermitian operators we find: 
$$
\mathrm{grad}f_A=\widetilde A \quad \mathrm{Ham} f_A=\widetilde {iA}, 
$$
where the vector fields associated with operators are defined by: 
$$
\widetilde A:\Hil_{\R}\to T\Hil_{R} \quad \psi\mapsto (\psi, A\psi) ,
$$
$$
\widetilde{iA}:\Hil_\R\to T\Hil_{R} \quad \psi\mapsto (\psi, JA\psi) .
$$

The  action of $U(\Hil)$ on $\Hil$ defines a momentum map 
$$
\mu:\Hil\to \mathfrak{u}^*(\Hil). 
$$
The fundamental vector fields associated with the operator $A$ is given by $%
\widetilde{iA}$ and the momentum map is such that 
$$
\mu(\psi)(\widetilde{iA})=\frac 12 \langle \psi, A\psi\rangle_{\Hil} 
$$
Thus we can write the momentum map from $\Hil_\R$ to $\mathfrak{u}^*(\Hil)$
as 
$$
\mu(\psi)=|\psi\rangle \langle \psi| 
$$

For Hermitian operators, linear functions on $\mathfrak{u}^*(\Hil)$ (i.e.
elements of the unitary algebra) are pulled-back to quadratic functions.
Therefore $\mu$ provides a symplectic realization of the Poisson manifold $%
\mathfrak{u}^*(\Hil)$ and the Ehrenfest picture is nothing but the
``pullback'' of the Heisenberg picture to the symplectic manifold $\Hil_\R$.
If we denote by $\hat A$ the linear function on $\mathfrak{u}^*(\Hil)$
associated with the element $-iA\in \mathfrak{u}(\Hil)$, we see immediately
that the momentum map relates the contravariant tensors $G$ and $\Omega$
(defined on $\Hil_\R$) with the linear contravariant tensors $R$ and $\Lambda
$ on $\mathfrak{u}^*(\Hil)$ corresponding to its Lie-Jordan brackets.

We have then the obvious definitions: 
$$
R(\xi)(\hat A, \hat B)=\langle \xi, [A,B]_+\rangle _{\mathfrak{u}^*}, \qquad
\Lambda (\xi)(\hat A, \hat B)=\langle \xi, [A,B]_-\rangle _{\mathfrak{u}^*}, 
$$
and together they form the complex tensor 
$$
(R+i\Lambda)(\xi)(\hat A, \hat B)=2\langle \xi, AB\rangle_{\mathfrak{u}^*}=%
\Tr \xi AB . 
$$
Clearly, 
$$
G(\mu^*\hat A, \mu^*\hat B)+i\Omega(\mu^*\hat A, \mu^*\hat B)=\mu^*(R(\hat
A, \hat B)+i\Lambda(\hat A, \hat B)). 
$$

Thus the momentum map provides also an unified view of the Schr\"odinger,
Ehrenfest and Heisenberg pictures. Clearly the Schr\"odinger vector field
being associated with the Hamiltonian function $\mu^*(\hat A)$ is just the usual
equation written as 
$$
i\frac d{dt}\psi=A\psi \,.
$$

The geometrical formulation of Quantum Mechanics we have presented shows that
the reduction procedure in the quantum setting may use most of the procedures
available from the classical setting. Of course now care must be used to deal
with the reduction of the nonlocal product. Again we may find that a reduced
$\star$--algebra need not be associated with a product defined on functions defined
on some ``quotient'' manifold. Thus whether or not the reduction procedure
commutes with the quantum-classical transition has to be considered an open
problem. 

\subsection{Pure states: the complex projective space}
The consideration that the probabilistic interpretation of Quantum Mechanics
requires state vectors to be normalized to one, i.e. $\langle \psi ,\psi
\rangle =1$, and that the probability density $\psi ^{\ast }(x,t)\psi
(x,t)$ is invariant under multiplication by a  phase, i.e. replacing $%
\psi $ with $e^{i\varphi} \psi $ does not alter the probabilistic
interpretation, imply that the carrier space of
``physical states'' is really the complex projective space $\mathbb{P}\Hil$
or the ray space $\mathcal{RH}$. If one considers the natural projection
from $\Hil-\{0\}$ to $\mathcal{RH}$: 
$$
\Hil-\{0\}\ni \psi \mapsto \pi (\psi )=\rho _{\psi }=\frac{|\psi \rangle
\langle \psi |}{\langle \psi ,\psi \rangle },
$$
one discovers that $\Hil-\{0\}$  can be considered as a
principal bundle over $\mathcal{RH}$ with group structure $\C_{0}=S^{1}\times \R_{+}$. The
infinitesimal generators of this group action, i.e. the corresponding
fundamental vector fields are $\Delta $ and $J(\Delta )$ (remember that $%
\Delta $ was the dilation vector field and $J$ the tensor representing the
complex structure of $\Hil$).

From the action of $U(\Hil)$ on $\Hil$, that we can write as 
$$
\psi \mapsto g\psi \quad g\in U(\Hil),
$$
we can introduce a ``projected'' action on $\mathcal{RH}$ given by 
$$
\rho\mapsto g\rho g^{-1}. 
$$
This action is transitive. In the particular case of the unitary evolution
operator (i.e. the one-parameter group of unitary transformations associated
with the Schr\"odinger equation (\ref{eq:schrodinger}), where we assume for
simplicity that the Hamiltonian does not depend on time), the evolution
written in terms of the elements of $\mathcal{RH}$ is written as 
$$
\rho(t)=\exp \left ( -\frac {iHt}{\hbar}\right )\rho(0)\exp \left ( \frac
{iHt}{\hbar}\right ). 
$$

This expression provides a solution of the von Neumann equation 
\begin{equation}
i\hbar \frac{d\rho }{dt}=[H,\rho ].  \label{eq:vonneumann}
\end{equation}

Thus this equation becomes another instance of the quantum equations of motion ,
in this case defined on $\mathcal{RH}$.

According to our previous treatment of the momentum map of the unitary
group, the ray space $\mathcal{RH}$ can be identified with a symplectic
leaf of $\mathfrak{u}^*(\Hil)$ passing through a rank-one projector 
$$\frac{|\psi \rangle \langle \psi|}{\langle \psi, \psi \rangle} \in \mathfrak{u}^*(\Hil)\,.$$
 Here we are interested in considering it as the complex projective
space obtained as a reduction of $\Hil-\{ 0\} $.

Now we would like to transfer the geometric objects we introduced on $%
\mathfrak{u}^*(\Hil)$ onto the ray space $\mathcal{RH}$. In particular the
structures we defined on the set of quadratic functions. As we are
interested now in functions which are projectable with respect to $\Delta$
and $J(\Delta)$, we are going to consider the functions associated with the
expectation values of the observables, i.e. 
$$
e_A(\psi)=\frac{\langle \psi, A\psi\rangle }{\langle \psi,\psi\rangle } 
$$

\begin{lemma}
These functions are invariant with respect the actions of $\Delta $ and $%
J(\Delta )$.
\end{lemma}

\begin{proof}
  The invariance with respect to dilations is obvious. To prove the invariance
  under $J(\Delta)$, it is useful to notice that this vector field is the Hamiltonian
  vector field (with respect the canonical symplectic structure) associated
  with the function 
$$f_I(\psi)=\frac 12 \langle \psi, I\psi\rangle \,.$$
 Hence it commutes with
  any quadratic function associated to an operator, because any operator
  commutes with the identity.
\end{proof}

This observation also shows that this example  may be considered to be the
Poisson reduction associated to the ideal generated by the functions $\langle \psi, \psi
\rangle -1$. The associated first class functions in the family of quadratic
ones are exactly given by $f_A(\psi)=\langle \psi, A\psi\rangle $, with $A$ a generic operator. In
this way the complex projective space provides an instance of Poisson
reduction. 

It is important to remark that in this construction we are not passing
through the submanifold $\Hil-\{ 0\} \supset \Sigma=\{ \psi\in \Hil-\{ 0\} |
\langle \psi,\psi\rangle =1\} $, as it is usually done in the definition of
the geometric description of the projective space. The reason
is that if we want to have the freedom to consider alternative Hermitian structures
on $\Hil$ (this would be the analog of the bi-Hamiltonian structures for classical
mechanical systems) we can not privilege a given one with respect to others. If
we change the  Hermitian structure, the submanifold $\Sigma$ would be different
while the corresponding projective space, as  a manifold, would not change .

It is now simple to understand why the tensors $G$ and $\Omega$, associated
with the Hermitian structure, will not be projectable objects. In spite of
this, we can turn them into projectable objects by introducing a conformal
factor: 
$$
\widetilde G=\langle \psi , \psi \rangle G,\qquad \widetilde \Omega=\langle
\psi, \psi \rangle \Omega. 
$$
But with this change, $\widetilde \Omega$ is no longer representing a
Poisson structure, but a Jacobi one, whose defining vector field is the
Hamiltonian vector field associated (with respect to the symplectic
structure) with the function $\frac 12 \langle \psi, \psi \rangle $, via $G$. The reduction of
this Jacobi algebra gives rise to the expected Poisson 
structure on the ray space $\mathcal{RH}$ \cite{IbdLeMar:1997}.

It is interesting to look at the particular form of these tensors when we
introduce adapted coordinates:

\begin{itemize}
\item  In complex coordinates the expression is 
$$\sum_{k}(z_{k}^{\ast
}z_{k})\sum_{l}\,\pd{\bar{z}_{l}}\otimes \pd{z_{l}}.$$

\item  If we use real coordinates, the principal bundle we have mentioned on
  the space of pure states will admit a connection one-form (which is
Hermitian) given by 
$$
\theta =\frac{\langle \psi ,d\psi \rangle }{\langle \psi ,\psi \rangle },
$$
which satisfies $\theta (J(\Delta ))=i$ and $\theta (\Delta )=1$; and
takes the form 
$$
\frac{(q-ip)d(q+ip)}{q^{2}+p^{2}}=\frac{qdq+pdp}{q^{2}+p^{2}}+i\frac{qdp-pdq%
}{q^{2}+p^{2}}
$$
with 
$$
\Delta =q\frac{\partial }{\partial q}+p\frac{\partial }{\partial p},\qquad
J(\Delta )=q\frac{\partial }{\partial p}-p\frac{\partial }{\partial q}.
$$
\end{itemize}

Once we have found symmetric and skew-symmetric tensors on the ray space we
can invert them and find the associated covariant form. When pulled-back to $%
\Hil-\{ 0\} $ these tensors may be represented by 
\begin{equation}  \label{eq:connection}
\frac{\langle d\psi, d\psi\rangle}{\langle \psi, \psi\rangle}-\frac{\langle
\psi, d\psi\rangle\langle d\psi, \psi\rangle}{\langle \psi, \psi\rangle^2}.
\end{equation}

This presentation shows very clearly that by changing the Hermitian structure
we also change the connection one form $\theta$, while $\Delta$ and $J$ remain unchanged if
we do not change the complex structure. The curvature two form of this
connection represents a symplectic structure on the ray space $\mathcal{RH}$
and is usually considered as a starting point to deal with geometric phases.

We consider the differential $d_J$ associated with the $(1,1)$--tensor $J$,
defined (see \cite{MoFeVecMarRu:1991,RoSpaVi:2006}) on functions and one--forms
as 
$$
(d_Jf)(X)=df(JX) \qquad d_J\beta (X,Y)=L_{JX}\beta(Y)-L_{JY}\beta(X)-\beta((J[X,Y]);
$$
for $\beta$ a one-form and $X$ and $Y$ vector fields and extended naturally to
higher order forms. With this 
we find that the K\"ahlerian two form $\omega$ may be written as 
$$
\omega=dd_J\log (\scalar{\psi, \psi}),
$$
while the translational invariant two form on $\Hil$ would be $dd_J(\scalar{\psi, \psi})$.

In this expression we see that $\log(\scalar{\psi, \psi})$ represents the K\"ahler
potential on $\Hil_\R$ and depends on the chosen Hermitian structure. It is not
projectable on the ray-space, while the two form we associate with it will be
the pull-back of a two form on $\mathcal{RH}$.

Before closing this section we notice that by taking convex combinations of
our pure states, rank-one projectors, we can generate the whole set of
density states. If, on the other hand, we consider real combinations, we
generate the full $\mathfrak{u}^*(\Hil)$ space. Therefore it is possible to
derive Heisenberg picture from the von Neumann description.

From our description in terms of geometrical Quantum Mechanics it should be
clear that the equivalence of the various pictures is naturally presented in
our generalized reduction procedure.

Another comment is in order. The reduction procedures within  Quantum Mechanics
are most effective when they are formulated in a  way such that the classical
limit may be naturally considered in the chosen formalism. We believe that this
may be considered as an indication that Quantum Mechanics should be formulated
in a way that in some form it incorporates the so called ``correspondence principle''.


\end{document}